\documentclass[conference]{IEEEtran}

\ifCLASSOPTIONcompsoc
  \usepackage[nocompress]{cite}
\else
  \usepackage{cite}
\fi

\ifCLASSINFOpdf
  \usepackage[pdftex]{graphicx}
\else
  \usepackage[dvips]{graphicx}
\fi

\usepackage{array}

\ifCLASSOPTIONcompsoc
  \usepackage[caption=false,font=normalsize,labelfont=sf,textfont=sf]{subfig}
\else
  \usepackage[caption=false,font=footnotesize]{subfig}
\fi

\usepackage{stfloats}

\usepackage{url}

\usepackage{booktabs} %
\usepackage{xspace}
\usepackage{mathtools}
\usepackage{calligra}
\usepackage{amsmath}
\usepackage{amssymb}
\usepackage{amsthm}
\usepackage{xcolor}
\usepackage{booktabs}
\usepackage{tabularx}
\usepackage{multirow}
\usepackage{arydshln}
\usepackage{bm}
\usepackage{nicefrac}
\usepackage{hyperref}
\usepackage{balance}

\usepackage[utf8]{inputenc}
\usepackage[T1]{fontenc}

\newcommand{\object}{\ensuremath{o}\xspace}
\newcommand{\secret}{\ensuremath{s}\xspace}
\newcommand{\objectspace}{\ensuremath{\mathbb{O}}\xspace}
\newcommand{\secretspace}{\ensuremath{\mathbb{S}}\xspace}
\newcommand{\examplespace}{\ensuremath{\secretspace\times\objectspace}\xspace}
\newcommand{\priors}{\ensuremath{\pi}\xspace}
\newcommand{\prior}[1]{\ensuremath{\pi(#1)}\xspace}
\newcommand{\classifier}{\ensuremath{f}\xspace}
\newcommand{\classifierspace}{\ensuremath{\mathcal{F}}\xspace}

\newcommand{\channel}{\ensuremath{\mathcal{C}}\xspace}   %
\newcommand{\bayesrisk}{\ensuremath{R^{*}}\xspace}
\newcommand{\bayesriskcond}{\ensuremath{r^{*}}\xspace}
\newcommand{\riskcond}[1]{\ensuremath{r(#1)}\xspace}
\newcommand{\errorguesspriors}{\ensuremath{R^{\priors}}\xspace}
\newcommand{\distance}{\ensuremath{d}\xspace}

\newcommand{\knn}{\ensuremath{k_n}-NN\xspace}
\DeclareMathOperator*{\argmax}{argmax}
\usepackage{pifont}%

\newcommand{\meleakage}{{{\rm ME}}\xspace}

\newcommand{\bleau}{F-BLEAU\xspace}
\newcommand{\gowalla}{\texttt{Gowalla}\xspace}
\newcommand{\leakwatch}{LeakWatch\xspace}
\newcommand{\leakiest}{leakiEst\xspace}

\DeclareMathAlphabet{\mathcalligra}{T1}{calligra}{m}{n}
\DeclareFontShape{T1}{calligra}{m}{n}{<->s*[2.2]callig15}{}
\newcommand{\riskfunction}{\ensuremath{\ell\xspace}}

\newtheorem{theorem}{Theorem}
\newtheorem{definition}{Definition}

\begin{document}
\title{F-BLEAU: Fast Black-box Leakage Estimation}

\newif\ifanonymous
\anonymousfalse

\ifanonymous\else
\author{
	\IEEEauthorblockN{Giovanni Cherubin}
	\IEEEauthorblockA{EPFL\\
		giovanni.cherubin@epfl.ch}
	\and
	\IEEEauthorblockN{Konstantinos Chatzikokolakis}
	\IEEEauthorblockA{University of Athens\\
		kostasc@di.uoa.gr}
	\and
	\IEEEauthorblockN{Catuscia Palamidessi}
	\IEEEauthorblockA{INRIA, École Polytechnique\\
		catuscia@lix.polytechnique.fr}
}

\fi

\maketitle
\thispagestyle{plain}
\pagestyle{plain}

\begin{abstract}
	We consider the problem of measuring how much a system
	reveals about its secret inputs.
	We work in the black-box setting: we assume no prior
	knowledge of the system's internals, and we run
	the system for choices of secrets and measure its leakage from the respective outputs.
	Our goal is to estimate the Bayes risk, 
	from which one can derive some of the most
	popular leakage measures (e.g., min-entropy
	leakage).
	
	The state-of-the-art method for estimating these
	leakage measures is the frequentist paradigm,
	which approximates the system's internals by
	looking at the frequencies of its inputs and outputs.
	Unfortunately, this does not scale for
	systems with large output spaces, where it would require
	too many input-output examples.
	Consequently, it also cannot be applied to systems
	with continuous outputs (e.g., time side channels, network traffic).
	
	In this paper, we exploit an analogy between
	Machine Learning (ML)  and black-box leakage estimation
	to show that the Bayes risk  
	of a system can be estimated by using a class
	of ML methods: the  \textit{universally consistent}
	learning rules; these rules can exploit patterns in the
	input-output examples to improve the estimates' convergence,
	while retaining formal optimality guarantees.
	We focus on a set of them,
	the nearest neighbor rules;
	we show that they
	significantly reduce
	the number of black-box queries required
	for a precise estimation
	whenever nearby outputs tend to be produced by
	the same secret;
	furthermore, some of them can
	tackle systems with continuous outputs.
	We illustrate the applicability of
	these techniques on both synthetic and real-world data,
	and we compare them with the state-of-the-art tool,
	\leakiest, which is based on the frequentist approach.
\end{abstract}

\IEEEpeerreviewmaketitle

\section{Introduction}

\begin{table*}[ht!]
	\caption{Number of examples required for convergence of the estimates.
		``X'' means an estimate did not converge.}
	\centering
	\renewcommand{\arraystretch}{1.3}
	\begin{tabular}{lcccc}
		\multicolumn{1}{l}{System}& {Dataset} & frequentist &    NN    &  \knn   \\
		\toprule

		Random                   & 100 secrets, 100 obs. & \textbf{10\,070} & \textbf{10\,070} & \textbf{10\,070} \\

		Geometric ($\nu=0.1$)  & 100 secrets, 10K obs. & 35\,016 & \textbf{333} & 458 \\
		Geometric ($\nu=0.02$) & 100 secrets, 10K obs. & 152\,904 & 152\,698 & \textbf{68\,058}\\
		Geometric ($\nu=2$)    & 10K secrets, 1K obs. & 95\,500 & \textbf{94\,204} & 107\,707 \\
		
		Multimodal Geometric ($\nu=0.1$)    & 100 secrets, 10K obs. & 44\,715 & \textbf{568} & 754 \\
		
		Spiky (contrived example) & 2 secrets, 10K obs. & \textbf{22\,908} & 29\,863 & 62\,325 \\
		~\vspace{-1.3em}\\ 
		\hline 
			Planar Geometric $\nu=2$ & \gowalla  checkins in San Francisco area & X  & X  & {\bf19\,948}\\
			Laplacian $\nu=2$ & " &N/A  & X  & {\bf19\,961} \\
			Blahut-Arimoto $\nu=2$ & " & 1\,285  & {\bf1\,170}  & 1\,343  \\
		\bottomrule
	\end{tabular}
	\vspace{1ex}
	
	The proposed tool, \bleau, is the combination of frequentist, NN, and \knn estimates,
	as an alternative to the frequentist paradigm.\\
	\label{tab:comparison}
\end{table*}

Measuring the information leakage of a system is one of the founding
pillars of security.
From side-channels to biases in random number generators, quantifying
how much information a system leaks about its secret inputs is
crucial for preventing adversaries from exploiting it;
this has been the focus of intensive research efforts in the areas of privacy and  of quantitative information flow (QIF).
Most approaches in the literature are based on the white-box approach, which consists in  
calculating analytically the channel matrix of the system,  constituted by the  conditional probabilities of the outputs given the secrets,    
and then computing the  desired
leakage measures 
(for instance, mutual information \cite{Clark:05:JLC}, min-entropy leakage \cite{Smith:09:FOSSACS}, or $g$-leakage~\cite{Alvim:12:CSF}).
However, 
while one typically has white-box access to the system they
want to secure, 
determining a system's leakage analytically is
often impractical, due to the size or complexity
of its internals, or to the presence of unknown factors. 
These obstacles led to 
investigate methods for   measuring a
system's leakage in a black-box manner.

Until a decade ago, the  most popular measure of leakage was 
Shannon mutual information (MI). However, in his seminal paper~\cite{Smith:09:FOSSACS} 
Smith showed that MI is not
appropriate to represent a realistic attacker, and proposed a notion of  
leakage based on R\'enyi min-entropy (\meleakage) instead. 
Consequently, in this paper we consider the general problem of estimating the Bayes risk of a system,
which is the smallest error achievable by an adversary at predicting its secret inputs
given the outputs.
From the Bayes risk one can derive several leakage measures, including \meleakage
and the additive and multiplicative
leakage~\cite{Braun:09:MFPS}.
These measures are considered by the QIF community among the most fundamental notions of leakage.

To the best of our knowledge, the only existing 
approach for the black-box estimation of the Bayes risk
comes from a
classical statistical technique,
which refer to as the
{\em frequentist} paradigm. 
The idea is to run the system repeatedly
on chosen secret inputs, and then count the
relative frequencies of the secrets and respective
outputs
so to estimate
their
joint probability distribution; 
from this distribution, it is then possible to
compute estimates of the desired leakage measure.
\leakwatch~\cite{ChothiaKN14:esorics} and \leakiest~\cite{Chothia:13:CAV},
two well-known tools for black-box leakage estimation,
are  applications of this principle.

Unfortunately, the frequentist approach does
not always  scale for real-world problems:
as the number of possible  input  and
output  values of the channel matrix increases, the number
of examples required for this method to
converge becomes too large to gather.
For example, \leakwatch requires a number of examples
that is much larger than the product of the size of
input and output space.
For the same reason, this method cannot be used for systems with continuous outputs; 
indeed,
it cannot even be
formally constructed in such a case.

\subsection*{Our contribution}
In
this paper, we show that machine learning (ML) methods
can provide the necessary scalability to black-box
measurements, and yet maintain formal guarantees
on their estimates.
By observing
a fundamental equivalence between ML and black-box leakage
estimation,
we show that any ML rule from a certain class 
(the {\em universally consistent} rules) can be used to estimate
with arbitrary precision
the leakage of a system.
In particular, we study
rules based on the nearest neighbor principle
-- namely, Nearest Neighbor (NN) and \knn, which exploit a metric
on the output space to achieve a
considerably faster convergence than frequentist
approaches. In \autoref{tab:comparison} we summarize the number of examples necessary for the 
estimators to converge, for the various systems considered in the paper.
We focus on nearest neighbor methods, among
the existing universally consistent rules, because:
i) they are simple to reason about,
and ii) we can identify the class of systems
for which they will excel,
which happens whenever the distribution is
somehow regular with respect to a metric on the output
(e.g., time side channels, traffic analysis, and
most mechanisms used for privacy).
Moreover, some of these methods can tackle directly
systems with continuous output.

We evaluate these estimators on synthetic data,
where we know the true distributions and
we can determine exactly when the estimates converge.
Furthermore, we use them
for measuring the leakage in a real dataset of users' locations,
defended with three state-of-the-art mechanisms: 
two geo-indistinguishability mechanisms (planar geometric and planar Laplacian)\cite{andres2013geo},
and the method by Oya et al.~\cite{Oya:17:CCS},
which we refer to as the Blahut-Arimoto mechanism. 
Crucially, the planar Laplacian is real-valued, which
\knn methods 
can tackle out-of-the box, but the frequentist method
cannot.
Results in both synthetic and real-world data
		show our methods give a strong advantage
		whenever there is a notion of metric in the output that can be exploited.
Finally, we compare  our methods with \leakiest on 
the problem of estimating the leakage of
European passports~\cite{chothia2010traceability,Chothia:13:CAV},
and on the location privacy data.

As a further evidence of their practicality,
we use them in  Appendix~\ref{appendix:exponentiation}
to measure the leakage of a time side channel in
a hardware implementation
of finite field exponentiation.

\subsection*{No Free Lunch}
A central takeaway of our work is that, while all
the estimators we study (including the frequentist
approach) are asymptotically optimal in the number of
examples, none of them
can guarantee on its finite sample performance;
indeed, no estimator can.
This is a consequence of the No Free Lunch
theorem in ML~\cite{wolpert1996lack}, which informally states that
all learning rules are equivalent
among the possible distributions of data.
This rules out the existence of an
optimal estimator.

In practice, this means that we should always evaluate
several estimators, and select the one that converged faster.
Fortunately, our main finding
(i.e., \textit{any universally consistent ML rule is a leakage estimator})
adds a whole new class
of estimators, which one can use in practical
applications.

We therefore propose a tool, \bleau
(Fast Black-box Leakage Estimation AUtomated),
which computes nearest neighbor and frequentist
estimates, and selects the one converging faster.
We release it as
Open Source software\footnote{\url{https://github.com/gchers/fbleau}.},
and we hope in the future to extend it to
support several more estimators based on UC ML rules. %

\subsection*{Nearest Neighbor rules}

Nearest neighbor
rules excel whenever there is a notion of metric on the output space, 
and the output distributions is ``regular'' (in the sense that it does not change too abruptly between two neighboring points).
We expect this to be the case for several real-world systems,
such as: side channels whose output is time, an electromagnetic signal, or power consumption;
for traffic analysis on network packets; and for geographic location data.
Moreover, most mechanisms used in privacy and QIF use smooth noise distributions.
Suitable applications may also come from recent
attacks to ML models, such as model inversion~\cite{fredrikson2015model}
and membership inference~\cite{shokri2017membership}.

Furthermore, we observe that even when there is no metric, or when the output distribution is irregular, 
(e.g., a system whose internal distribution has been randomly
sampled), these rules are equivalent to the frequentist
approach.
Indeed, the only case we observe when they are misled is when the
system is crafted so that the metric contributes \textit{against}
classification (e.g., see ``Spiky'' example in \autoref{tab:comparison}).

\section{Related Work}

Chatzikokolakis et al.~\cite{chatzikokolakis2010statistical}
introduced methods for measuring the leakage of a deterministic
program in a black-box manner;
these methods worked by collecting a large number of inputs and
respective outputs, and by estimating the underlying probability
distribution accordingly; this is what we refer
to as the frequentist paradigm.
A fundamental development of their work by Boreale and Paolini~\cite{boreale2014formally}
showed that, in the absence of significant a priori information about
the output distribution,
no estimator does better than the exhaustive enumeration of the input domain.
In line with this work, \autoref{sec:nfl} will show that, as a consequence of the No Free Lunch
theorem in ML, no leakage estimator can claim to converge faster than
any other estimator for all distributions.

The best known tools for black-box estimation of leakage
		measures based on the Bayes risk (e.g., min-entropy) are 
\leakiest~\cite{Chothia:11:CSF,Chothia:13:CAV} and
\leakwatch~\cite{ChothiaKN14:esorics,chothia2013probabilistic}, both
based on the frequentist paradigm.
The former also allows a zero-leakage test for systems
with continuous outputs.
In \autoref{sec:comparison} we provide a comparison of \leakiest with our proposal.

Cherubin~\cite{cherubin2017bayes}
used the guarantees of nearest neighbor learning rules
for estimating the information leakage
(in terms of
		the Bayes risk)
of defenses against website fingerprinting attacks
in a black-box manner.

Shannon mutual information (MI) is the main alternative to the Bayes risk-based notions of 
leakage in the QIF literature.
Although there is a relation between MI  and Bayes risk~\cite{Santhi:06:ACCCC}, the corresponding  models of attackers are very different: 
the first corresponds to an attacker who can try infinitely
many times to guess the secret, while the second has only
one try at his disposal~\cite{Smith:09:FOSSACS}. 
Consequently, MI and Bayes-risk measures, such as ME,
can give very different results:
Smith~\cite{Smith:09:FOSSACS} shows two programs that have almost
the same MI, but one has an ME  several orders of magnitude larger than the other one;
conversely, there are examples of two programs such that
ME is $0$ for both, while the MI is $0$ in one case and strictly positive (several bits) in the other one. 

In the black-box literature, MI is usually computed by
using Kernel Density Estimation,
which although only
guarantees asymptotic optimality
under smoothness assumptions on the distributions.
On the other hand, the ML literature offered developments in this area:
Belghazi et al.~\cite{belghazi2018mine} proposed an MI lower bound
estimator based on deep neural networks, and proved its consistency
(i.e., it converges to the true MI value asymptotically).
Similarly, other works constructed MI variational lower bounds~\cite{chen2016infogan,chen2018learning}.

\section{Preliminaries}
\label{sec:preliminaries}

\begin{table}
	\centering
	\caption{Symbols table.}
	\begin{tabular}{cl}
		Symbol & Description \\ \toprule		
		$s \in \secretspace$ & A secret\\
		$o \in \objectspace$ & An object/black-box output\\
		$(s, o) \in \examplespace$ & An example\\
		$(\priors, \channel_{\secret,\object})$ & A system, given a set of
				priors $\priors$ and channel matrix $\channel$\\
		$\mu$ & Distribution induced by a system on $\examplespace$\\
		$f: \secretspace \mapsto \objectspace$ & A classifier\\
		$\riskfunction$ & Loss function
			w.r.t. which we evaluate a classifier\\
		$R^f$ & The expected misclassification error of a classifier $f$\\
		$R^*$ & Bayes risk\\
		\bottomrule
	\end{tabular}
	\label{tab:symbols-table}
\end{table}

We define a system,
and show that its leakage can be expressed in terms of the Bayes
risk.
We then introduce ML notions,
which we will later use
to estimate the Bayes risk.

\subsection{Notation}
We consider a system $(\priors, \channel_{\secret,\object})$,
that associates a secret input
$\secret$ to an observation (or object) $\object$ in a possibly
randomized way.
The system is defined by a set of prior probabilities
$\prior{\secret} := P(\secret)$, $\secret \in \secretspace$, and a
channel matrix $\channel$ of size $|\secretspace| \times |\objectspace|$,
for which $\channel_{\secret,\object} := P(\object | \secret)$
for $\secret \in \secretspace$ and $\object \in \objectspace$.
We call \examplespace the example space.
We assume the system does not change over time;
for us, $\secretspace$ is finite, and
$\objectspace$ is finite unless otherwise stated.

\subsection{Leakage Measures}
The state-of-the-art in  QIF is represented by  the leakage measures based on $g$-{\em vulnerability}, 
a family whose most representative member is min-vulnerability~\cite{Smith:09:FOSSACS}, the complement of the Bayes risk.
This paper
is concerned with finding tight
estimates of the Bayes risk,
which can then be used to estimate the
appropriate leakage measure.

\paragraph{Bayes risk}

The Bayes risk, \bayesrisk, is the error of the optimal (idealized) classifier
for the task of predicting a secret $s$ given an observation $o$
output by a system.
It is defined with respect to
a loss function $\riskfunction: \secretspace \times \secretspace \mapsto \mathbb{R}_{\geq 0}$,
where $\riskfunction(s, s')$ is the risk of an adversary predicting $s'$ for
an observation $o$, when its
actual secret is $s$.
We focus on the 0-1 loss function,
$\riskfunction(s, s') := I(s \neq s')$, taking value $1$ if $s \neq s'$,
$0$ otherwise.
The Bayes risk of a system $(\priors, \channel_{\secret,\object})$
is defined as:
\begin{equation}
\bayesrisk:= 1 - \sum_{\object \in \objectspace}
\max_{\secret \in \secretspace} \channel_{\secret, \object} \prior{s} \,.
\label{eq:bayes-from-channel}
\end{equation}

\paragraph{Random guessing}
A baseline for evaluating a system
is the error committed by an idealized
adversary who knows priors but has no access to
the channel, and who's best strategy
is to always output
the secret with the highest prior.
We call the error of this adversary random guessing error:
\begin{equation}
	\errorguesspriors := 1 - \max_{\secret \in \secretspace} \prior{\secret} \,.
\end{equation}

\subsection{Black-box estimation of \bayesrisk}

This paper is concerned with estimating the Bayes risk
given $n$ examples sampled from the joint distribution $\mu$ on $\examplespace$
generated by $(\priors, \channel_{\secret,\object})$.
By running the system $n$ times on secrets
$s_1,\ldots , s_n \in \secretspace$, chosen according to $\priors$,
we generate a sequence of corresponding outputs $o_1,\ldots , o_n$, thus forming a  {\em training set}\footnote{In line with
	the ML literature,
	we call the training or test ``set'' what is technically
	a multiset;
	also, we loosely
	use the set notation ``$\{\}$'' for both
	 sets and multisets when  the nature of the object
	is clear from the context.} of examples
$\{(o_1, s_1), ..., (o_n, s_n)\}$.
From these data, we aim to make an estimate
close to the real Bayes
risk.%

\subsection{Learning Rules}

We introduce
ML rules (or, simply, {\em learning rules}), which are
algorithms
for selecting a classifier given a set of training examples.
In this paper, we will use
the error of some ML rules
as an estimator of the Bayes risk.

Let $\classifierspace := \{ \classifier \mid \classifier: \objectspace \mapsto \secretspace \}$
be a set of classifiers.
A learning rule is a possibly randomized algorithm that, given a training
set $\{(o_1, s_1), ..., (o_n, s_n)\}$,
returns a classifier $\classifier \in \classifierspace$,
with the goal of minimizing the expected loss $\mathbb{E}\,\riskfunction(\classifier(o), s)$ for a
new example $(\object, \secret)$ sampled from $\mu$~\cite{vapnik2013nature}.
In the case of  the 0-1 loss function, the expected loss coincides with the \emph{expected probability of error} (\emph{expected error} for short), and if $\mu$ is generated by a system $(\priors, \channel_{\secret,\object})$, then the expected error of a classifier $\classifier : \objectspace \mapsto \secretspace$ is:
\begin{equation}\label{eq:error}
R^{\classifier} = 1 - \sum_{\object \in \objectspace}
\channel_{\classifier(\object), \object} \prior{f(\object)}  
\end{equation}
where $\classifier(\object)$ is the secret predicted for object $\object$.
If $\objectspace$ is infinite (and $\mu$ is continuous) the summation
is replaced by an integral.

\subsection{Frequentist estimate of \bayesrisk}
\label{sec:frequentist}

The frequentist paradigm~\cite{chatzikokolakis2010statistical}
for measuring the leakage of
a channel consists in estimating the probabilities
$\channel_{\secret,\object}$ by counting their frequency
in the training data $(\object_1, \secret_1), ..., (\object_n, \secret_n)$:
\begin{equation}
P(\object | \secret) \approx \hat{\channel}_{\secret, \object} :=
\frac{|i : \object_i = \object, \secret_i = \secret|}{|i : \secret_i = \secret|} \,.
\end{equation}

We can obtain the frequentist error from
\autoref{eq:error}: 
\begin{equation}R^{\mathit{Freq}} =  1- \sum_o  C_{f^\mathit{Freq}(o),o} \; \pi(f^\mathit{Freq}(o)) \end{equation}
where $f^\mathit{Freq}$ is the frequentist classifier, namely:
\begin{equation}
	f^\mathit{Freq}(o) =
	\begin{cases}
		\argmax_s (\hat{\channel}_{\secret, \object} \, \hat{\pi}(\secret)) &
			\mbox{if $o$ in training data}\\
		\argmax_s \hat{\pi}(\secret) & \mbox{otherwise} \,,
	\end{cases}
\end{equation}
where $\hat{\pi}$ is estimated from the examples:
$\hat{\pi}(\secret)=\nicefrac{|i : \secret_i = \secret|}{n}$.

Consider a finite example space $\examplespace$.	
Provided with enough examples, the frequentist approach always
converges:
clearly,
$\hat{\channel} \rightarrow \channel$ as $n \rightarrow \infty$,
because events' frequencies converge to
their probabilities by the Law of Large Numbers.

However, there is a fundamental issue with this approach.
Given a training set $\{(o_1, s_1), ..., (o_n, s_n)\}$,
the frequentist classifier can tell something meaningful
(i.e., better than random guessing)
for an object $\object \in \objectspace$, only as long as $o$
appeared in the training set;
but, for very large systems (e.g., those with a large object space),
the probability of observing an example for each object
becomes small, and the frequentist
classifier approaches \textit{random guessing}.
We study this matter further in \autoref{subsec:random-system}
and Appendix~\ref{appendix:frequentist}.

\section{No Free Lunch In Learning}
\label{sec:nfl}

The frequentist approach performs well only for objects
it has seen in the training data;
in the next section, we will introduce estimators
that aim to provide good predictions even for unseen objects.
However, we shall first answer an important question:
is there an estimator that is ``optimal'' for all systems?

A negative answer to this question is given by
the so-called ``No Free Lunch'' (NFL) theorem
by Wolpert~\cite{wolpert1996lack}.
The theorem is formulated for the expected
	loss of a learning rule on {\em unseen} objects
	(i.e., that were not in the training data),
	which is referred to as
the off-training-set (OTS) loss.

\begin{theorem}[No Free Lunch]
	Let  $A_1$ and $A_2$ be two learning rules,  $\riskfunction$ 
	  a cost function, and $\mu$ a distribution on $\examplespace$. We indicate by 
	$\mathbb{E}_i(\riskfunction  \mid  \mu, n)$ the OTS loss of $A_i$ given $\mu$ and $n$, 
	where the expectation is computed over all the possible training sets of size $n$ 
	sampled from $\mu$. 
	Then, if we take the uniform average among all possible distributions $\mu$, we have
	\begin{equation}\mathbb{E}_1(\riskfunction \mid \mu, n) = \mathbb{E}_2(\riskfunction \mid \mu, n) \,.\end{equation}
\end{theorem}

Intuitively, the NFL theorem says that, if all distributions
(and, therefore, all channel matrices) are equally likely, then
all learning algorithms are equivalent.
Remarkably, this holds for {\em any} strategy, even if
one of the rules is random guessing.

An important implication of this for our purposes is that
for every two learning rules $A$ and $B$ there will always exist
some system for which rule $A$ converges faster than $B$,
and vice versa there will be a system for which $B$ outperforms
$A$.

From the practical perspective of black-box security, this demonstrates that
we should always test several estimators
and select the one that converges faster.
Fortunately, the connection between ML and black-box security
we highlight in this paper results in the discovery of
a whole class of new estimators.

\section{Machine Learning Estimates of the Bayes Risk}
\label{sec:ml-estimates}

In this section, we define the notion of a {\em universally consistent}
learning rule, and show that the error of
a classifier selected according to such a rule can be used 
for estimating the Bayes risk.
Then, we introduce various universally consistent
rules based on the nearest neighbor principle.

Throughout the section, we use interchangeably 
a system $(\priors, \channel_{\secret,\object})$ and its corresponding joint distribution $\mu$ on $\examplespace$. 
Note that there is a one-to-one correspondence between them.

\subsection{Universally Consistent Rules}
Consider a distribution $\mu$ and  
a learning rule $A$ selecting a classifier
$\classifier_n \in \classifierspace$
according to $n$ training examples sampled from  $\mu$.
Intuitively, as the available training data increases,
we would like the expected error 
of  $\classifier_n$ for a new example $(o, s)$ sampled
from
$\mu$ to be minimized (i.e., to get close the Bayes risk).
The following definition captures this intuition.

\begin{definition}[Consistent Learning Rule]
Let $\mu$ be a distribution on ${\examplespace}$ and 
let $A$ be a learning rule. 
Let $\classifier_n \in \classifierspace$ be  a classifier selected by $A$ 
using $n$ training examples sampled from  $\mu$.
Let $(\priors, \channel_{\secret,\object})$ be the system corresponding to $\mu$, and let 
$R^{\classifier_n}$ be the expected error of $\classifier_n$, as 
defined by~\eqref{eq:error}. 
We say that $A$ is consistent if $R^{\classifier_n} \rightarrow \bayesrisk$
	as $n \rightarrow \infty$.
\end{definition}
The next definition strengthens this property, by asking
the rule to be consistent for all distributions:
\begin{definition}[Universally Consistent (UC) Learning Rule]
	A learning rule is universally consistent if it is consistent for any
	distribution $\mu$ on \examplespace.
\end{definition}

By this definition, the expected error of  a classifier selected according to a
universally consistent rule is also an estimator of the Bayes
risk, since it converges to \bayesrisk as $n\rightarrow\infty$.

In the rest of this section we introduce Bayes risk estimates
based on universally consistent nearest neighbor rules;
they are summarized in \autoref{tab:estimates} together with
their guarantees.

\begin{table}
	\centering
	\caption{Estimates' guarantees as $n \rightarrow \infty$}
	\renewcommand{\arraystretch}{1.3}
	\begin{tabular}{lcc}
		\toprule
		Method & Guarantee & Space $\objectspace$ \\
		\midrule
		frequentist & $\rightarrow \bayesrisk$ & finite\\
		NN & $\rightarrow \bayesrisk$ & finite\\
		\knn & $\rightarrow \bayesrisk$ & infinite, $(\distance, \objectspace)$ separable\\
		NN Bound$^+$ & $\leq \bayesrisk$ & infinite, $(\distance, \objectspace)$ separable\\
		\bottomrule
	\end{tabular}\\
	$^+$NN Bound is discussed in Appendix~\ref{appendix:more-tools}.
	\label{tab:estimates}
\end{table}

\subsection{NN estimate}

The Nearest Neighbor (NN) is one of the simplest ML classifiers:
given a training set and a new object $o$, it predicts
the secret of its closest training observation (\textit{nearest neighbor}).
It is defined both for finite and infinite object spaces,
although it is UC only in the first case.

We introduce a formulation of NN, which can
be seen as an extension of the frequentist approach,
that takes into account \textit{ties} (i.e., neighbors that are equally close to
the new object $o$), and which guarantees consistency
when $\objectspace$ is finite.

Consider a training set $\{(o_1, s_1), ..., (o_n, s_n)\}$,
an object  $o$,
and   a distance metric
$\distance: \objectspace \times \objectspace \mapsto \mathbb{R}_{\geq0}$.
The NN classifier predicts a secret for $o$ by taking
a majority vote over the set of secrets whose objects
have the smallest distance to $o$.
Formally, let $I_\mathit{min}(o) = \{ i \mid  \distance(o, o_i) = \min_{j=1...n} \distance(o, o_j)  \}$
and define:
\begin{equation}
	\mathit{NN}(o) = s_{h(o)}
\end{equation}
where
\begin{equation}
	h(o) = \argmax_{i\in{I_\mathit{min}(o)}} |\{ j \in I_\mathit{min}(o)\mid s_j=s_i\}| \,.
\end{equation}

We show that NN is universally consistent
for finite $\examplespace$.

\begin{theorem}[Universal consistency of NN]
	Consider a distribution on $\examplespace$, where
	$\secretspace$ and $\objectspace$ are finite.
	Let $R^{\mathit{NN}}_n$ be the expected error of the NN classifier
	for a new observation $o$.
	As the number of training examples $n \rightarrow \infty$:
	\begin{equation}R^{\mathit{NN}}_n \rightarrow \bayesrisk.\end{equation}
\end{theorem}

\begin{proof}[Sketch proof]
	For an observation $\object$ that appears in the
	training set, the NN classifier is equivalent to the frequentist
	approach. For a finite space $\examplespace$, as $n \rightarrow \infty$,
	the probability that the training set contains all $\object \in \objectspace$
	approaches $1$.
	Thus, the NN rule is asymptotically (in $n$) equivalent to
	the frequentist approach,  which means its error also converges
	to \bayesrisk.
\end{proof}

\subsection{\knn estimate}

Whilst NN guarantees universal consistency in
finite example spaces,
this does not hold for infinite
\objectspace.
In this case, we can achieve  universal consistency
with the k-NN classifier,
an extension of NN, for appropriate choices of the parameter $k$.

The k-NN classifier takes a majority
vote among the secrets of its neighbors.
Breaking ties in the k-NN definition requires more care than with NN.
In the literature, this is generally done via strategies that
add randomness or arbitrariness to the choice (e.g., if two neighbors
have the same distance, select the one with the smallest index
in the training data)~\cite{devroye2013probabilistic}.
We use a novel tie-breaking strategy, which takes into
account ties, but gives more importance to the closest neighbors.
In early experiments, we observed this strategy had a faster convergence
than standard approaches.

Consider a training set $\{(o_1, s_1), ..., (o_n, s_n)\}$,
an object to predict $o$, and some metric
$\distance: \objectspace \times \objectspace \mapsto \mathbb{R}_{\geq0}$.
Let $o_{(i)}$ denote the $i$-th closest object to $o$, and $s_{(i)}$
its respective secret.
If ties do not occur after the $k$-th neighbor (i.e., if
$\distance(o, o_{(k)}) \neq \distance(o, o_{(k+1)})$), then k-NN outputs the
most frequent among the secrets of the first $k$ neighbors:
\begin{equation}k\mbox{-}\mathit{NN}(o) = s_{h(o)} \end{equation}
where
\begin{equation}h(o) = \argmax_{i=1, ..., k} |\{ j\in I_\mathit{min}(o) \mid s_{(j)}=s_{(i)}\}| \,.\end{equation}
If ties exist after the $k$-th neighbor, that is,
for $k' \leq k < k''$:
\begin{equation}\distance(o, o_{(k')}) = ... = \distance(o, o_{(k)}) = ... = \distance(o, o_{(k'')}) \,,\end{equation}
we proceed as follows.
Let $\hat{s}$ be the most frequent secret in
$\left\{s_{(k')}, ..., s_{(k'')}\right\}$;
k-NN predicts the most frequent secret in
the following multiset, truncated at the tail to have size $k$:
\begin{equation*}
	s_{(1)}, s_{(2)}, ..., s_{(k'-1)}, \hat{s}, \hat{s}..., \hat{s} \,.
\end{equation*}

We now define \knn, a universally consistent learning rule
that selects a k-NN classifier for a training set of $n$ examples
by choosing $k$ as a function of $n$.

\begin{definition}[\knn rule]
	Given a training set of $n$ examples, the \knn rule selects
	a k-NN classifier, where $k$ is chosen
	such that $k_n \rightarrow \infty$ and $k_n / n \rightarrow 0$ as $n \rightarrow \infty$.
\end{definition}

Stone proved that \knn is universally
consistent~\cite{stone1977consistent}:%

\begin{theorem}[Universal consistency of the $k_n$-NN rule]
	\label{thm:knn-consistency}
	Consider a probability distribution $\mu$
	on the example space $\examplespace$,
	where $\mu$ has a density.
	Select a distance metric $\distance$ such that
	$(\distance, \objectspace)$ is separable\footnote{A separable space
		is a space containing a countable dense subset;
		e.g., finite spaces
		and
		the space of $q$-dimensional vectors $\mathbb{R}^q$
		with Euclidean metric.}.
	Then the expected error of the $k_n$-NN rule
	converges to \bayesrisk as $n \rightarrow \infty$.
\end{theorem}

This holds for any distance metric.
In our experiments, we will use the Euclidean distance, and
we will evaluate \knn rules for
$k_n = \log{n}$ (natural logarithm)
and $k_n = \log_{10}{n}$.

The ML literature is rich of UC rules and other useful
tools for black-box security; we list some of them in
Appendix~\ref{appendix:more-tools}.

\section{Evaluation on Synthetic Data}
\label{sec:synthetic-experiments}

In this section we evaluate our estimates on several
synthetic systems for which the channel matrix is known.
For each system, we sample $n$ examples from its distribution,
	and then compute the estimate on the whole
	object space as in \autoref{eq:error};
	this is possible because \objectspace is finite.
Since for synthetic data we know the real Bayes risk,
we can measure how many examples are required for
the convergence of each estimate.
We do this as follows:
let $R^f_n$ be an estimate of \bayesrisk, trained on
a dataset of $n$ examples.
We say the estimate $\delta$-converged to \bayesrisk after $n$ examples if
its relative change from \bayesrisk is smaller than $\delta$:
\begin{equation}
	\frac{\left| R^f_n - \bayesrisk \right|}{\bayesrisk} < \delta \,.
\end{equation}
While relative change has the advantage of taking into account the magnitude
of the compared values, it is not defined when the denominator is $0$;
therefore, when $\bayesrisk \approx 0$ (\autoref{tab:all-systems}),
we verify convergence with the absolute change:
\begin{equation}
	\left| R^f_n - \bayesrisk \right| < \delta \,.
\end{equation}

\begin{table}
	\centering
	\caption{Synthetic systems.}
	\begin{tabular}{lcccc}
		Name     & Privacy parameter   & $|\secretspace|$ & $|\objectspace|$ & \bayesrisk \\ \toprule		
		Geometric & 1.0 &       100        &       10K        &  $\sim 0$  \\
		Geometric & 0.1 &       100        &       10K        &   0.007    \\
		Geometric & 0.01 &       100        &       10K        &   0.600    \\
		Geometric & 0.2 &       100        &        1K        &   0.364    \\
		Geometric & 0.02 &       100        &       10K        &   0.364    \\
		Geometric & 0.002 &       100        &       100K       &   0.364    \\
		
		Geometric & 2 & 100K & 100K & 0.238 \\
		Geometric & 2 & 100K & 10K & 0.924 \\ %
		
		Multimodal & 1.0 & 100 & 10K  & 0.450\\
		Multimodal & 0.1 & 100 & 10K & 0.456\\
		Multimodal & 0.01 & 100 & 10K & 0.797\\
				
		Spiky & N/A & 2 & 10K & 0\\
		
		Random   & N/A   &       100        &       100        &   0.979\\
		\bottomrule
	\end{tabular}
	\label{tab:all-systems}
\end{table}

The systems used in our experiments are briefly discussed in
this section, and summarized in \autoref{tab:all-systems};
we detail them in Appendix~\ref{appendix:systems}.
A uniform prior is assumed in all cases.

\subsection{Geometric systems}

We first consider systems generated by adding geometric noise to the secret, 
one of the typical mechanisms used to implement differential privacy~\cite{Dwork:06:ICALP}.
Their channel matrix has the following form:
\begin{equation}
\channel_{\secret, \object} = P(o \mid s) =
\lambda \exp\left(- \nu {\mid g(s) - o \mid}\right) \,,
\end{equation}
where $\nu$ is a privacy parameter, $\lambda$ a normalization
factor, and $g$ a function $\secretspace \mapsto \objectspace$;
a detailed description of these systems is given in
in Appendix~\ref{appendix:systems}.

We consider the following three parameters: 
\begin{itemize}
	\item the privacy parameter $\nu$,
	\item the ratio $\nicefrac{|\objectspace|}{|\secretspace|}$, and
	\item the size of the secret space $|\secretspace|$.
\end{itemize}
We vary each of these parameters one at a time, to isolate their effect on the convergence rate. 

\subsubsection{Variation of the privacy parameter $\nu$}

We fix  $|\secretspace|=100$, $|\objectspace|=10$K, and we consider  three cases
$\nu=1.0$, $\nu=0.1$ and $\nu=0.01$. The
results for the estimation of the Bayes risk and the convergence rate 
are illustrated in \autoref{fig:geometric-system-vary-noise} and
\autoref{tab:geometric-system-vary-noise} respectively. 
In the table, results are reported for  convergence level
$\delta \in \{0.1, 0.05, 0.01, 0.005\}$; an ``X'' means a particular estimate did not converge
within 500K examples; a missing row for a certain $\delta$ means no estimate
converged.

The results indicate that
the nearest neighbor methods have a much faster
convergence than the standard frequentist approach,
particularly when dealing with large systems.
The reason is that  geometric systems have a regular behavior with respect to the Euclidean metric,
which can be exploited by NN and \knn to make good predictions
for unseen objects.

\begin{figure*}[h!]
	\centering
	\includegraphics[width=0.65\textwidth]{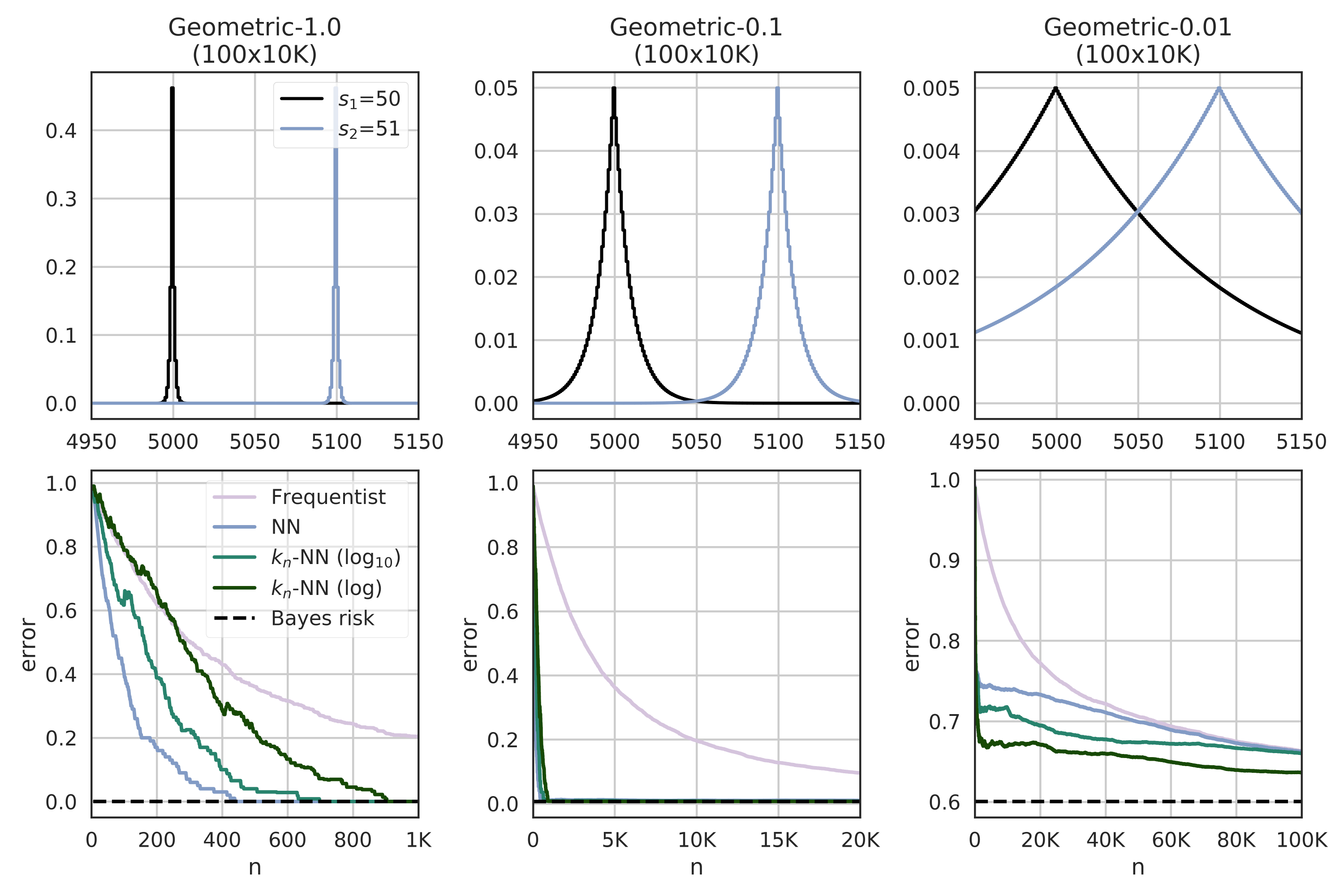}
	\caption{Estimates' convergence for geometric systems when varying their privacy parameter $\nu$.
		The respective distributions are shown in the top figure
		for two adjacent secrets $s_1 \sim s_2$.}
	\label{fig:geometric-system-vary-noise}
\end{figure*}
\begin{table}
	\centering
	\caption{Convergence of the estimates 
		when varying $\nu$, fixed
		$|\secretspace| \times |\objectspace| = 100 \times 10K$}
	\begin{tabular}{lccccc}
		&          &       &    & \multicolumn{2}{c}{\knn}\\
		\cline{5-6}
		System & $\delta$ & Freq. & NN & $\log_{10}$ & $\log$ \\
		\toprule
		
		\multirow{2}{5em}{{\bf Geometric\\$\bm \nu$ = 1.0}}
		& 0.1 & 1\,994 & \textbf{267} & 396 & 679\\
		& 0.05 & 4\,216 & \textbf{325} & 458 & 781\\
		& 0.01 & 19\,828 & \textbf{425} & 633 & 899\\
		& 0.005 & 38\,621 & \textbf{439} & 698 & 904\\
		
		\vspace{-0.3em}\\
		\multirow{2}{5em}{{\bf Geometric \\$\bm \nu$ = 0.1}}
		& 0.1 & 18\,110 & \textbf{269} & 396 & 673\\
		& 0.05 & 35\,016 & \textbf{333} & 458 & 768\\
		& 0.01 & 127\,206 & \textbf{439} & 633 & 899\\
		& 0.005 & 211\,742 & 4\,844 & \textbf{698} & 904\\
		\vspace{-0.3em}\\

		\multirow{2}{5em}{{\bf Geometric \\$\bm \nu$ = 0.01}}
		& 0.1 & 105\,453 & 103\,357 & 99\,852 & \textbf{34\,243}\\
		& 0.05 & 205\,824 & 205\,266 & 205\,263 & \textbf{199\,604}\\
		\bottomrule
	\end{tabular}
	\label{tab:geometric-system-vary-noise}
\end{table}

\subsubsection{Variation of the ratio $\nicefrac{|\objectspace|}{|\secretspace|}$}
Now we fix  $|\secretspace| =100$, 
and we consider  three cases
$\nicefrac{|\objectspace|}{|\secretspace|}=10$, $\nicefrac{|\objectspace|}{|\secretspace|}=100$, and $\nicefrac{|\objectspace|}{|\secretspace|}=1$K.
(Note that we want to keep
the ratio $\nicefrac{\nu}{\Delta_g}$ fixed, see Appendix~\ref{appendix:systems};
as a consequence $\nu$ has to vary: we set $\nu$  to $0.2$, $0.02$, and $0.002$, respectively.)
Results in \autoref{fig:geometric-system-vary-ratio} and
\autoref{tab:geometric-system-vary-ratio}
show how the nearest neighbor methods
become much better than the frequentist approach as $|\objectspace|/|\secretspace|$ increases. 
This is because the larger the object space, the larger the number of unseen objects at the moment of classification,
and the more the frequentist approach has to rely on random guessing.
The  nearest neighbor methods are not that much affected because they can rely on the proximity to outputs already classified.

\begin{figure*}
	\centering
	\includegraphics[width=0.65\textwidth]{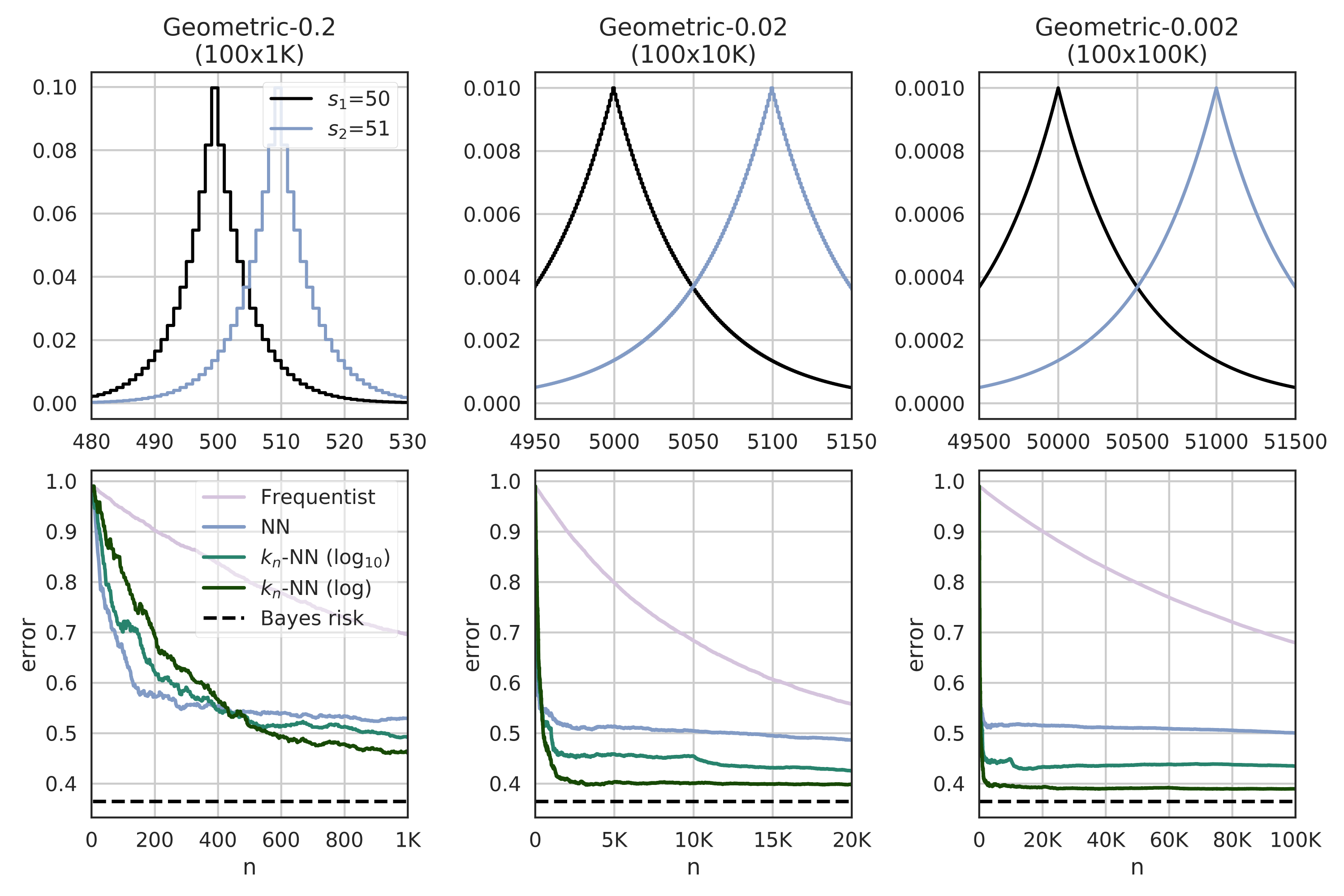}
	\caption{Estimates' convergence for geometric systems when varying the
		ratio $|\objectspace|/|\secretspace|$.
		The respective distributions are shown in the top figure
		for two adjacent secrets $s_1 \sim s_2$.}
	\label{fig:geometric-system-vary-ratio}
\end{figure*}

\begin{table}
	\centering
	\caption{Convergence of the estimates 
	 when varying
		$|\objectspace|/|\secretspace|$.}
	\begin{tabular}{lccccc}
		&          &       &    & \multicolumn{2}{c}{\knn}\\
		\cline{5-6}
		System & $\delta$ & Freq. & NN & $\log_{10}$ & $\log$ \\
		\toprule
		\multirow{3}{5em}{{\bf Geometric 100x1K \\$\bm \nu$ = 0.2}}
		& 0.1 & 8\,679 & 8\,707 & 7\,108 & \textbf{2\,505}\\
		& 0.05 & 14\,823 & 14\,853 & 14\,853 & \textbf{7\,673}\\
		& 0.01 & \textbf{51\,694} & 60\,796 & 60\,796 & 60\,796\\
		& 0.005 & \textbf{71\,469} & \textbf{71\,469} & \textbf{71\,469} & \textbf{71\,469}\\
		
		\vspace{-0.3em}\\

		\multirow{3}{5em}{{\bf Geometric 100x10K \\$\bm \nu$ = 0.02}}
		& 0.1 & 85\,912 & 85\,644 & 71\,003 & \textbf{11\,197}\\
		& 0.05 & 152\,904 & 152\,698 & 151\,153 & \textbf{68\,058}\\
		~\\
		
		\vspace{-0.3em}\\
		\multirow{3}{5em}{{\bf Geometric 100x100K \\$\bm \nu$ = 0.002}}
		& 0.1 & X & X & 413\,974 & \textbf{2\,967}\\
		~\\
		~\\
		\bottomrule
	\end{tabular}
	\label{tab:geometric-system-vary-ratio}
\end{table}

\subsubsection{Case $|\secretspace| \geq |\objectspace|$}
		We fix  $\nu = 2$, and we consider two cases:
		$|\secretspace| = |\objectspace|$
		and $|\secretspace| > |\objectspace|$.
		It should be noted that the formulation of geometric systems prohibits the number of secrets
		to exceed the number of outputs; for this reason, in the system
		$|\secretspace| > |\objectspace|$ some secrets are associated with the
		same distribution over the output space (Appendix~\ref{appendix:systems}).
		
		The results in \autoref{fig:geometric-system-vary-secrets-objects}
		and \autoref{tab:geometric-system-vary-secrets-objects}
		indicate that NN and frequentist are mostly equivalent: this is because
		they both need to observe at least one example for each secret.
		\knn rules, on the other hand, show poor performances, due to
		the fact that they would need at least $k_n$ examples for each
		secret.
		A natural extension of our work is to look at notions of metric
		also in the secret space for improving convergence.

\begin{figure}[h!]
	\centering
	\includegraphics[width=0.85\linewidth]{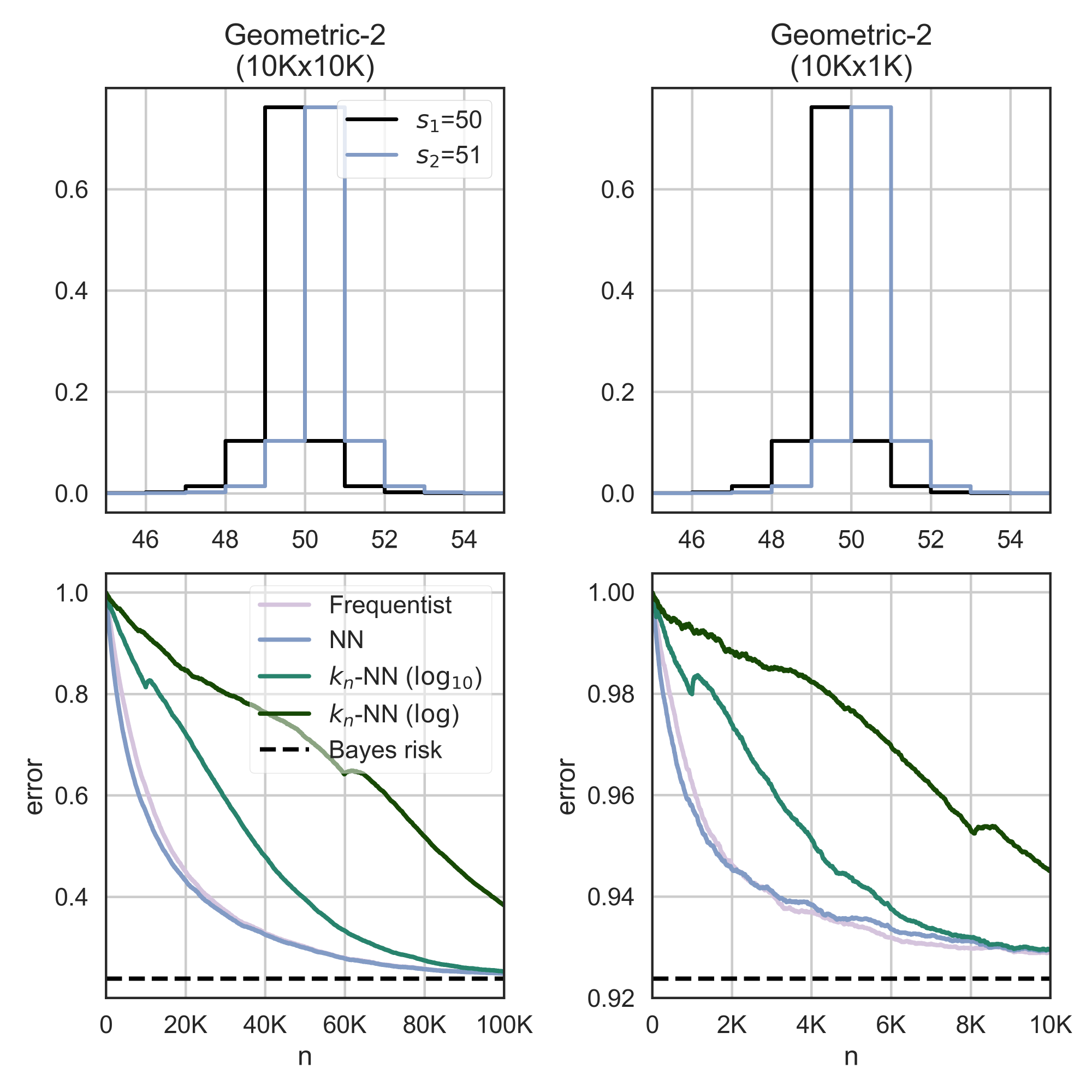}
	\caption{Estimates' convergence for geometric systems when $|\secretspace| \geq |\objectspace|$.
		The distributions are shown in the top figure
		for two adjacent secrets $s_1 \sim s_2$.
		In the case $|\secretspace| > |\objectspace|$ (right) there are
		$10 = \nicefrac{10K}{1K}$ identical distributions that
		coincide on $s_1$, and $10$ identical distributions on $s_2$.}

	\label{fig:geometric-system-vary-secrets-objects}
\end{figure}
	
\begin{table}[h!]
	\centering
	\caption{Convergence of the estimates 
		when 
		$|\secretspace| \geq|\objectspace|$, $\nu = 2$.}	
	\begin{tabular}{lccccc}
		&          &       &    & \multicolumn{2}{c}{\knn}\\
		\cline{5-6}
		System & $\delta$ & Freq. & NN & $\log_{10}$ & $\log$ \\
		\toprule	
		\multirow{3}{5em}{{\bf Geometric 10Kx10K}}
		& 0.1 & 74\,501 & \textbf{73\,085} & 88\,296 & 140\,618\\
		& 0.05 & 95\,500 & \textbf{94\,204} & 107\,707 & 155\,403\\
		& 0.01 & \textbf{137\,099} & 137\,348 & 144\,846 & 192\,014\\
		& 0.005 & \textbf{153\,370} & \textbf{153\,370} & 159\,075 & 203\,363\\
		
		\vspace{-0.3em}\\
		\multirow{3}{5em}{{\bf Geometric 10Kx1K}}
		& 0.1 & \textbf{5} & \textbf{5} & \textbf{5} & \textbf{5}\\
		& 0.05 & 721 & \textbf{514} & 2\,309 & 5\,977\\
		& 0.01 & \textbf{5\,595} & 6\,171 & 7\,330 & 12\,354\\
		& 0.005 & \textbf{10\,770} & 10\,797 & 11\,037 & 14\,575\\
		\bottomrule
	\end{tabular}
	\label{tab:geometric-system-vary-secrets-objects}
\end{table}

\subsection{Multimodal geometric system}
We now evaluate the estimators on systems with a multimodal
distribution.
In particular, we create multimodal geometric systems by
summing two geometric probability distributions, appropriately normalized
and shifted by some parameter.
We provide the details of this distribution in Appendix~\ref{appendix:systems}.

\begin{figure*}[h!]
	\centering
	\includegraphics[width=0.65\linewidth]{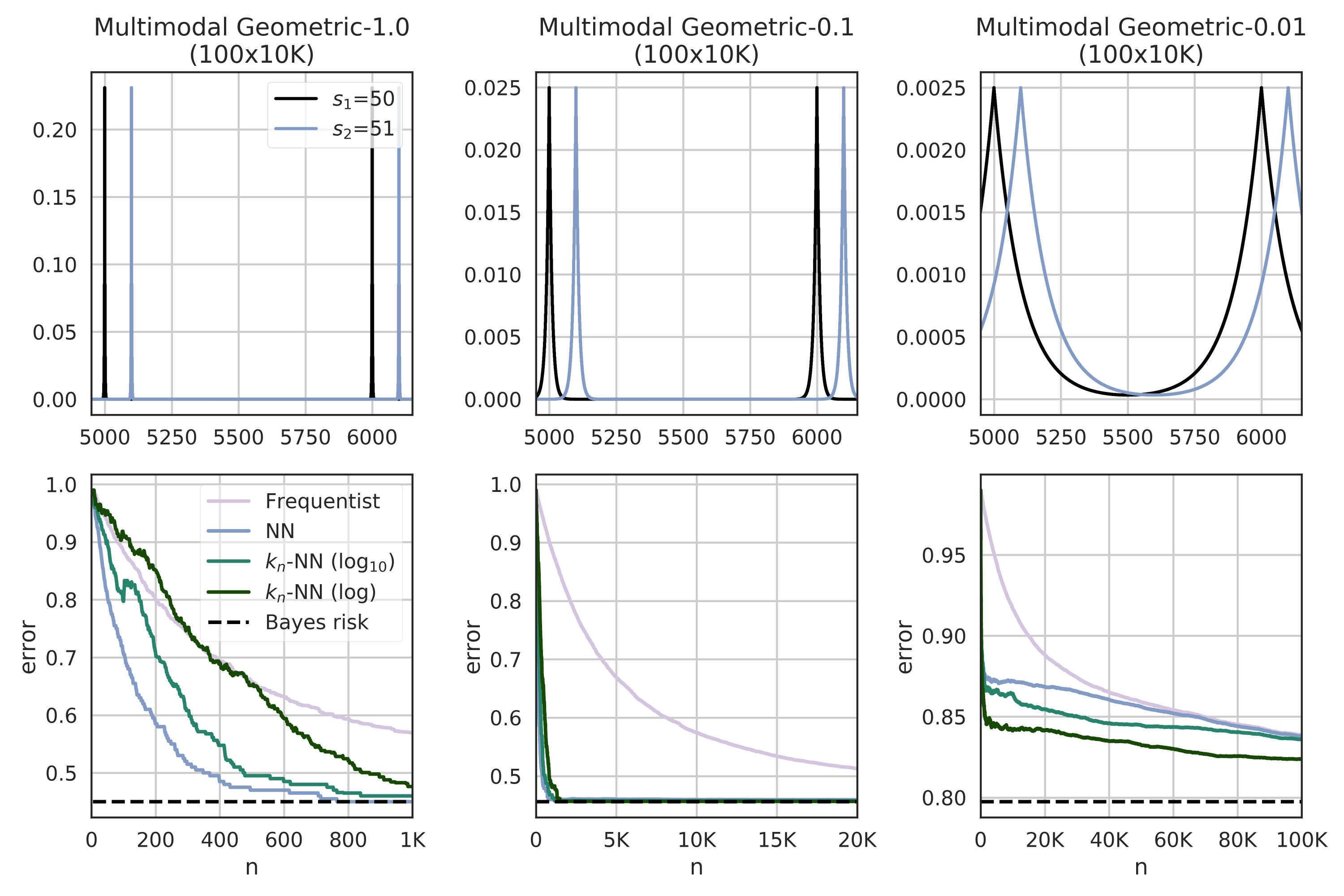}
	\caption{Estimates' convergence for multimodal geometric systems when varying
		the privacy parameter $\nu$.
		The distributions are shown in the top figure
		for two adjacent secrets $s_1 \sim s_2$.}
	\label{fig:multimodal-convergence}
\end{figure*}

\subsubsection{Evaluation}
Results are reported in \autoref{fig:multimodal-convergence}.
As expected, we observe that nearest neighbor rules improve
on the frequentist approach; the reason is that, even for
multimodal distributions, there exists a metric on the outputs
which they can exploit.
Detailed $\delta$-convergence results are
in Appendix~\ref{appendix:convergence-results}.

\subsection{Spiky system: When \knn rules fail}
\label{sec:knn-fail}

Nearest neighbor rules take advantage of the metric
on the object space
to improve their convergence considerably.
However, as a consequence of the NFL theorem,
there exist systems for which 
the frequentist approach outperforms NN and \knn.
Investigating the form of such systems
is important to understand when these
methods fail.

We craft one such system, the Spiky system,
where the metric misleads predictions.
The Spiky system
is such that
neighboring points are associated with different
secret.
This means that NN and \knn rules will tend to
predict the wrong secret, until enough examples are
available.
We detail its construction in Appendix~\ref{appendix:systems}.

\begin{figure}
	\centering
	\includegraphics[width=0.50\linewidth]{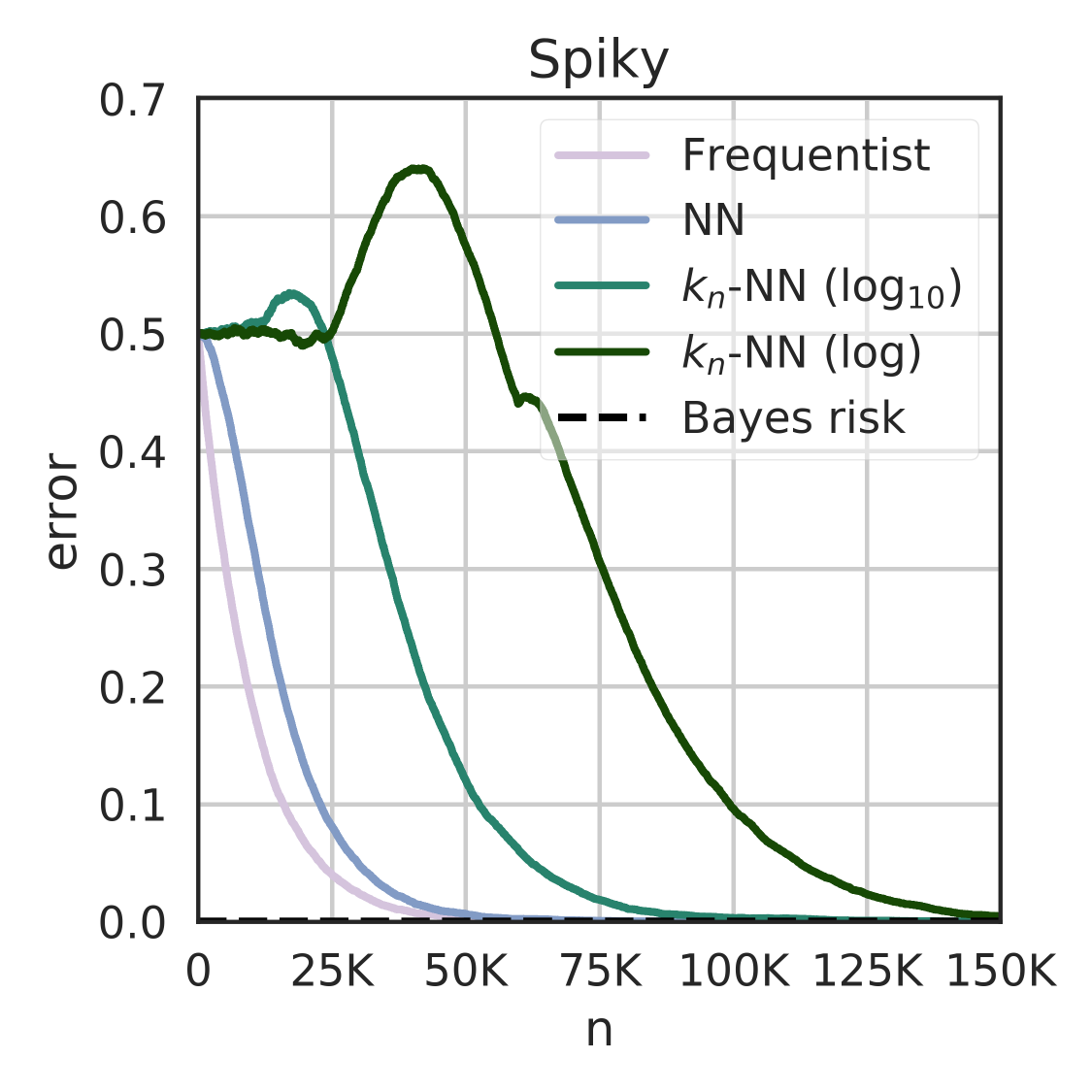}
	\caption{Estimates' convergence for a Spiky system (2x10K).}
	\label{fig:spiky-system}
\end{figure}

\paragraph{Discussion}
We conducted experiments for a Spiky system of size $|\objectspace| = 10K$.
Results in \autoref{fig:spiky-system} confirm the hypothesis:
nearest neighbor rules are misled for this
system.

Interestingly, while the NN estimate keeps decreasing as the number
of examples $n$ increases, there is a certain range of $n$'s where the \knn estimates become worse than random
guessing.
Intuitively, this is because   
when $n$ becomes larger than $|\objectspace|$, all elements in  
$\objectspace$ tend to be covered by the examples. For every $i\in \objectspace$ there are two 
neighbors, $i-1$ and $i+1$, that belong to the class opposite to the one of $i$, so if $k$ is not too small with respect to $n$, 
it is likely  that in the multiset of the $k$ closest neighbors of $i$, the number of $i-1$'s and $i+1$'s exceeds the 
number of $i$'s, which means that $i$ will be misclassified. 
As $n$ increases, however, the ratio between $k$ and the number of $i$'s in the examples tends to decrease 
(because
$\nicefrac{k}{n}\rightarrow 0$ as $n\rightarrow \infty$), hence at some point we will have enough 
$i$'s to win the majority vote in the $k$ neighbors
($i$'s are considered before than $i-1$'s and $i+1$'s, by
the nearest neighbor definition)
so $i$ will not be   misclassified anymore.

Concerning the comparison between the NN and frequentist estimates, we can do it analytically. 
We start by computing the expected error of the NN method 
on the spiky system in terms of the number of training examples $n$.
Let $T^n$ be a training set of examples of size $n$. Given a new object $i$,  let
us consider the NN estimate $r_n(i)$ of $\bayesriskcond$ for  $i$, i.e., 
the expected probability of error in the classification of $i$. 
This is the probability that the element $o$ closest to $i$ that appears in the training set has odd distance  from $i$ (i.e., $d(i,o) = 2\ell +1$, for some natural number $\ell$). Namely it is the probability that:
\begin{itemize}
	\item  $i$ is not in the training data but either ${i+1}$ or ${i-1}$ are, or
	\item $i,  {i\pm1},  {i\pm2}$ are not in the training data but either
	${i+3}$ or ${i-3}$ are, or
	\item \ldots etc.
\end{itemize}
Hence we have:
\begin{align}
\riskcond{i} =&\; P(d(i, o) = 2l+1) =\\
=&\, P(i \notin T^n, {i+1} \in T^n) + P(i \notin T^n, {i-1} \in T^n) + \ldots\\
=&\, 2\cdot \sum_{\ell=0}^{\nicefrac{q}{4}-1} a^{4\ell+1} (1-a),
\end{align}
where $a= (1-\nicefrac{1}{q})^n$  is the probability that an element $e\in\objectspace$ does not occur in any of the $n$ 
examples of the training set. 
(Thus $a^{4\ell+1}$ represents the probability that none of the elements $i, i\pm1,i\pm2,i\pm 2 s$, with $\ell= 2s$, appear in the training set, and  
$1-a$ represents the probability that the element $2s+1$ (resp.  $2s-1$ ) appears in the training set.) 
By using the result of the geometric series 
\begin{equation}
\sum_{t = 0}^{m} a^t = \frac{1-a^{m+1}}{1-a},
\end{equation}
we obtain:
\begin{equation}
r_n(i) =
2 a \frac{1-a^q}{(1+a^2)(1+a)}.
\end{equation}
Since we assume that the distribution on $\objectspace$ is uniform, we have $R^\mathit{NN}_n=r_n(i)$.

We want to study how the error estimate depends on the relative size of the training set with respect to the size of $\objectspace$. Hence,  let $x= \nicefrac{n}{q}$. 
Then we have $a = (1-\nicefrac{1}{q})^{qx}$, which, for large $q$, becomes $a \approx e^{-x}$. 
Therefore:
\begin{equation}
R^\mathit{NN}_x \approx
2 e^{-x} \frac{1-e^{-qx}}{(1+e^{-2x})(1+e^{-x})}.
\end{equation}
It is easy to see that $R^\mathit{NN}_x \rightarrow \nicefrac{1}{2}$ for $x\rightarrow 0$, and $R^\mathit{NN}_x \rightarrow 0$ for $x\rightarrow \infty$, as expected. 

Consider now the frequentist estimate $R^\mathit{Freq}_x$. 
In this case, given an element $i\in\objectspace$, the classification is done correctly if  $i$ appears in the training set. 
Otherwise, we do random guessing, which gives a correct or wrong classification with equal probability. Only the latter case contributes to the 
probability of error, hence the error estimate is half the probability expectation that $i$ does not belong to the training set:
\begin{equation}
R^\mathit{Freq}_x = \frac{1}{2}(1-\frac{1}{q})^n \approx  \frac{1}{2} e^{-x}
\end{equation}
Therefore, $R^\mathit{NN}_x $ is always above $R^\mathit{Freq}_x$.

\subsection{Random System}
\label{subsec:random-system}

In the previous sections, we have seen cases when our methods
greatly outperform the frequentist approach, and a contrived system
example for which they fail.
We now consider a system generated randomly to evaluate
their performances for an ``average'' system.

\paragraph{System description}
We construct
the channel matrix of a Random system
by drawing
its elements from the uniform distribution,
$\channel_{\secret, \object} \leftarrow^{\$} Uni(0,1)$, and
normalizing its rows.

\paragraph{Evaluation}
We consider a Random system with
$|\secretspace| = |\objectspace| = 100$
and count the number of examples required for
$\delta$-convergence, for many $\delta$'s.
\autoref{tab:random-system} reports the results.

\begin{table}[ht]
	\centering
	\caption{Random: examples required for $\delta$-convergence.}
	\renewcommand{\arraystretch}{1.3}
	\begin{tabular}{ccccc}
		&       &    & \multicolumn{2}{c}{\knn}\\
		\cline{4-5}
		$\delta$ & Freq. & NN & $\log_{10}$ & $\log$ \\
		\toprule
		0.05 & \textbf{5} & \textbf{5} & \textbf{5} & \textbf{5}\\
		0.01 & \textbf{82} & 139 & 202 & 500\\
		0.005 & \textbf{10\,070} & \textbf{10\,070} & \textbf{10\,070} & \textbf{10\,070}\\
		\bottomrule
	\end{tabular}
	\label{tab:random-system}
\end{table}
\begin{figure}
	\centering
	\includegraphics[width=0.5\linewidth]{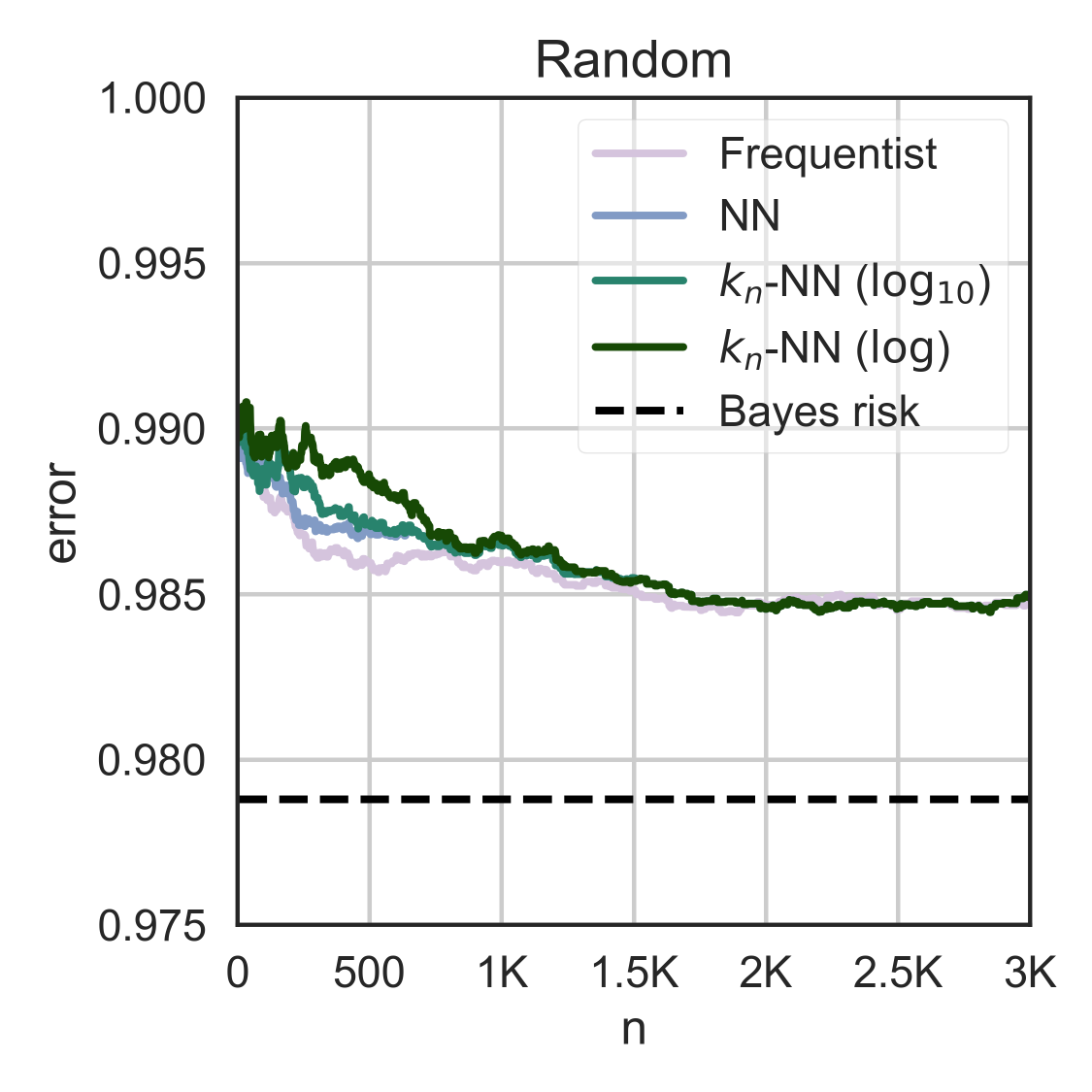}
	\caption{Estimates' convergence for a Random system ($100\times 100$).}
	\label{fig:random-system}
\end{figure}

The frequentist estimate is slightly better than NN and \knn for $\delta=0.01$.
However, for stricter convergence requirements
($\delta=0.001$), all the methods require the same (large) number of examples. 
\autoref{fig:random-system} shows  that indeed the methods begin to converge similarly
already after 1K examples.

\paragraph{Discussion}
\label{sec:synthetic-frequentist-discussion}

Results showed that nearest neighbor estimates require
significantly fewer examples than the frequentist approach
when dealing with medium or large systems;
however, %
they are generally
equivalent to the frequentist approach  in the case of small systems.

To better understand why this is the case, we provide
a crude approximation of the frequentist Bayes risk estimate.

\begin{equation}
R^{\mathit{Freq}}_n \approx \bayesrisk \left(1 - \left(1 - \frac{1}{|\objectspace|}\right)^n\right) +
\errorguesspriors\left(1 - \frac{1}{|\objectspace|}\right)^n \,.
\end{equation}

This approximation, derived and studied in Appendix~\ref{appendix:frequentist},
makes the very strong assumption that all objects are equally likely,
i.e.:
$P(o) = \frac{1}{|\objectspace|}$.
However, this is enough to give us an insight on
the performance of the frequentist approach:
$\left(1 - \frac{1}{|\objectspace|}\right)^n$ is the probability
that some object does not appear within a training set of
size $n$. This
  weighs the value of the frequentist
estimate between the optimal \bayesrisk, used when the
object appears in the training data, and random guessing
\errorguesspriors:
while the estimate
converges asymptotically to
the Bayes risk,
the probability of observing an object
-- often related to the the size $|\objectspace|$,
has a major influence on its convergence rate.

\section{Application to Location Privacy}
\label{sec:gowalla}

\begin{figure}[ht]
	\centering
	\includegraphics[width=0.5\linewidth]{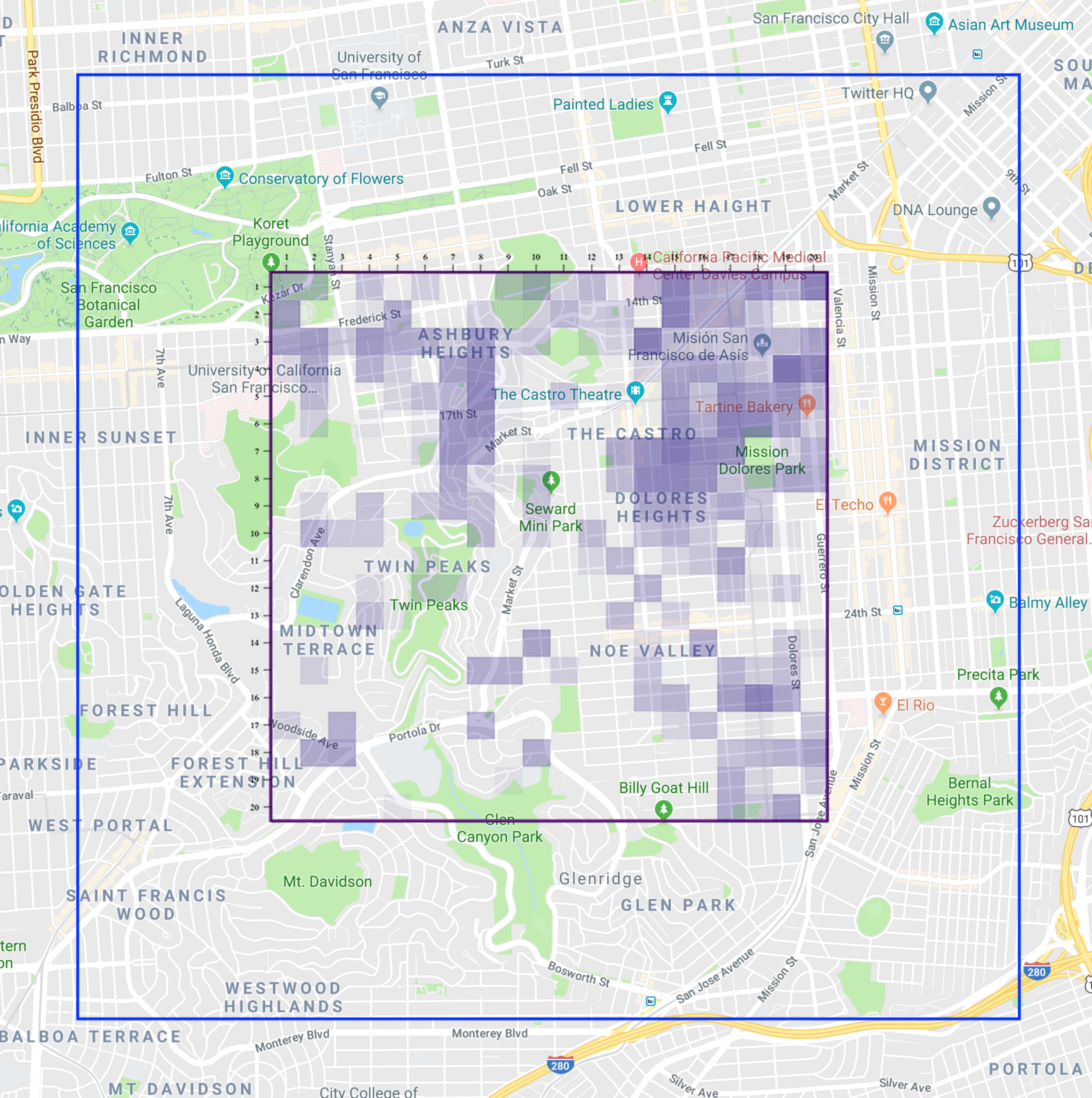}
	\caption{Area of San Francisco considered for the experiments. The input locations correspond to the inner square,
		the output locations to the outer one. The colored cells represent the distribution of the \gowalla checkins.}
	\label{fig:mapSF}
\end{figure}

We show that  \bleau can be successfully applied to estimate the degree of protection provided by  mechanisms 
such as those used in location privacy. 
Since the purpose of this paper is to evaluate the precision of  \bleau, we consider basic mechanisms 
for which  the Bayes risk can also be  computed directly.
Of course, the intended applications of  \bleau are mechanisms or situations where the  Bayes risk \emph{cannot} be computed directly, either because 
this is too complicated, or because of the presence of unknown factors. 
Examples abound;  for instance, the availability of additional information,  like the presence of points of interest (e.g., shops, churches), or geographical characteristics of the area (e.g., roads, lakes) can affect the Bayes risk in ways that are impossible to evaluate formally.  

We will consider the planar  Laplacian and the  planar geometric, 
which are the typical mechanisms used to obtain geo-indistinguishability~\cite{andres2013geo},
and one of the optimal mechanisms proposed by Oya et al.~\cite{Oya:17:CCS} as a refinement 
of the optimal mechanism by Shokri et al.~\cite{Shokri:12:CCS}. 
The construction of
the last
relies on an algorithm that was independently proposed by Blahut and by Arimoto
to solve  the information theory problem of achieving an
optimal trade-off between the  minimization of  the  distortion rate  and the
the mutual information~\cite{Cover:06:BOOK}.
From now on, we shall refer to this as the Blahut-Arimoto mechanism.
Note that the Laplacian is a continuous mechanism
(i.e., its outputs are on a continuous plane);
the other two are discrete.

In these experiments we also deploy the method that \bleau  uses in practice to compute the estimate of the Bayes risk:
 we first split the data into a training set and a hold-out set; then, for an increasing number of examples
$n=1, 2, ...$ we train the classifier on the first $n$ examples on the
training set, and then estimate its error on the hold-out set.%

\subsection{The \gowalla dataset}

We consider real location data from the \gowalla dataset~\cite{Gowalla,Cho:11:SIGKDD},
which contains users' checkins and their geographical location.
We use
a  squared area in San Francisco,
centered in the coordinates (37.755, -122.440), 
and extending for 1.5 Km in each direction. This input area corresponds to the inner (purple) square in \autoref{fig:mapSF}.
We discretize the input using a grid of $20 \times 20$ cells of size $150\times 150$ Sq m;
the secret space $\secretspace$ of the system thus consists of $400$ locations.
The  prior distribution on the secrets is derived from the  \gowalla checkins, and it
is represented in \autoref{fig:mapSF} by the different color intensities
on the input grid. 
The output area is represented in \autoref{fig:mapSF} by the outer (blue) square. It extends 1050 m (7 cells) more than the input square on every side.
We consider a larger area  for the output
because the planar Laplacian and Geometric
naturally expand outside the input square.\footnote{In fact these functions distribute the probability on the infinite plane, but on locations very distant from the origin the probability becomes negligible.}
Since the planar Laplacian is continuous, its  output domain  $\objectspace$ is constituted by all the points of the outer square. 
As for the planar Geometric and the Blahut-Arimoto mechanisms, which are discrete, we divide the output square in a grid of
$340 \times 340$
cells of size $15\times 15$ Sq m;
therefore, $|\objectspace| = 340 \times 340 = 115,600$.

\subsection{Defenses}

The  \textit{planar Geometric} mechanism is   characterized by a channel matrix $C_{s,o}$, representing the conditional probability of reporting the location $o$ when the true location is $s$:
\begin{equation}C_{s,o}= \lambda \exp\left(- \frac{\ln \nu}{100} d(s,o) \right) \,,\end{equation}
where $\nu$ is a parameter controlling the level of noise, $\lambda$ is a normalization factor, and $d$ is the Euclidean distance. %

The \textit{planar Laplacian} is defined by the same equation, except that $o$ belongs to a continuous domain, and the equation defines a probability density function. 

As for the \textit{Blahut-Arimoto}, it is obtained as the result of an iterative algorithm, whose definition can be found in
\cite{Cover:06:BOOK}.

\subsection{Results}
We evaluated the estimates' convergence as a
function of the number of training examples $n$ and for different
values of the noise level: $\nu = \{2, 4, 8\}$.
We randomly split the dataset ($100K$ examples) into
training ($75\%$) and hold-out ($25\%$) sets,
and then evaluated the convergence of the estimators on an increasing number of
training examples, $5, 6, ...75K$.

Results for the geometric noise (\autoref{fig:geometric-gowalla})
indicate faster convergence when $\nu$ is higher (which means less noise and lower Bayes risk), 
in line with the results for the synthetic systems of the previous section. 
In all  cases, the nearest neighbor methods outperform the frequentist one, as we expected given the presence of a large number of outputs.  
\autoref{tab:geometric} shows the number of examples required
to achieve $\delta$-convergence from the Bayes risk.
The symbol ``X'' means we did not achieve a certain level of approximation with 75K examples. 
\begin{figure*}[t]
	\centering
	\includegraphics[width=0.65\textwidth]{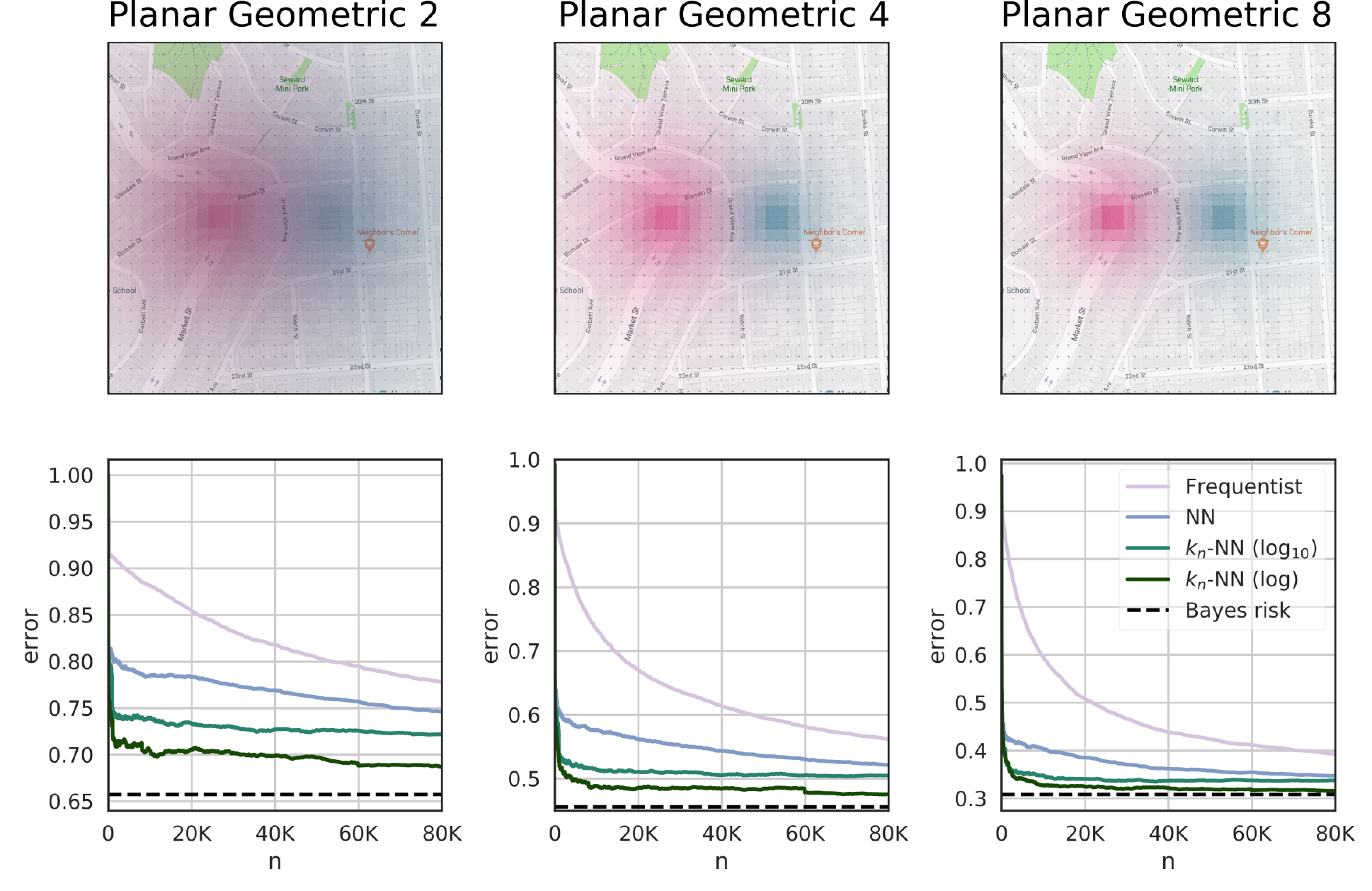}
	\caption{Estimates' convergence speed for the planar Geometric defense applied to the \gowalla dataset, for  $\nu = 2$, $\nu = 4$ and $\nu = 8$, respectively. Above each graph is represented the  distribution of the geometric noise for two adjacent input cells.}
	\label{fig:geometric-gowalla}
\end{figure*}

The corresponding results for the  Laplacian noise are shown in \autoref{fig:geometric-gowalla}  and in \autoref{tab:laplace}.
In this case, the frequentist approach is not applicable, but
the \knn rule can still approximate the Bayes risk for some
approximation levels.

\begin{figure*}[t]
	\centering
	\includegraphics[width=0.65\textwidth]{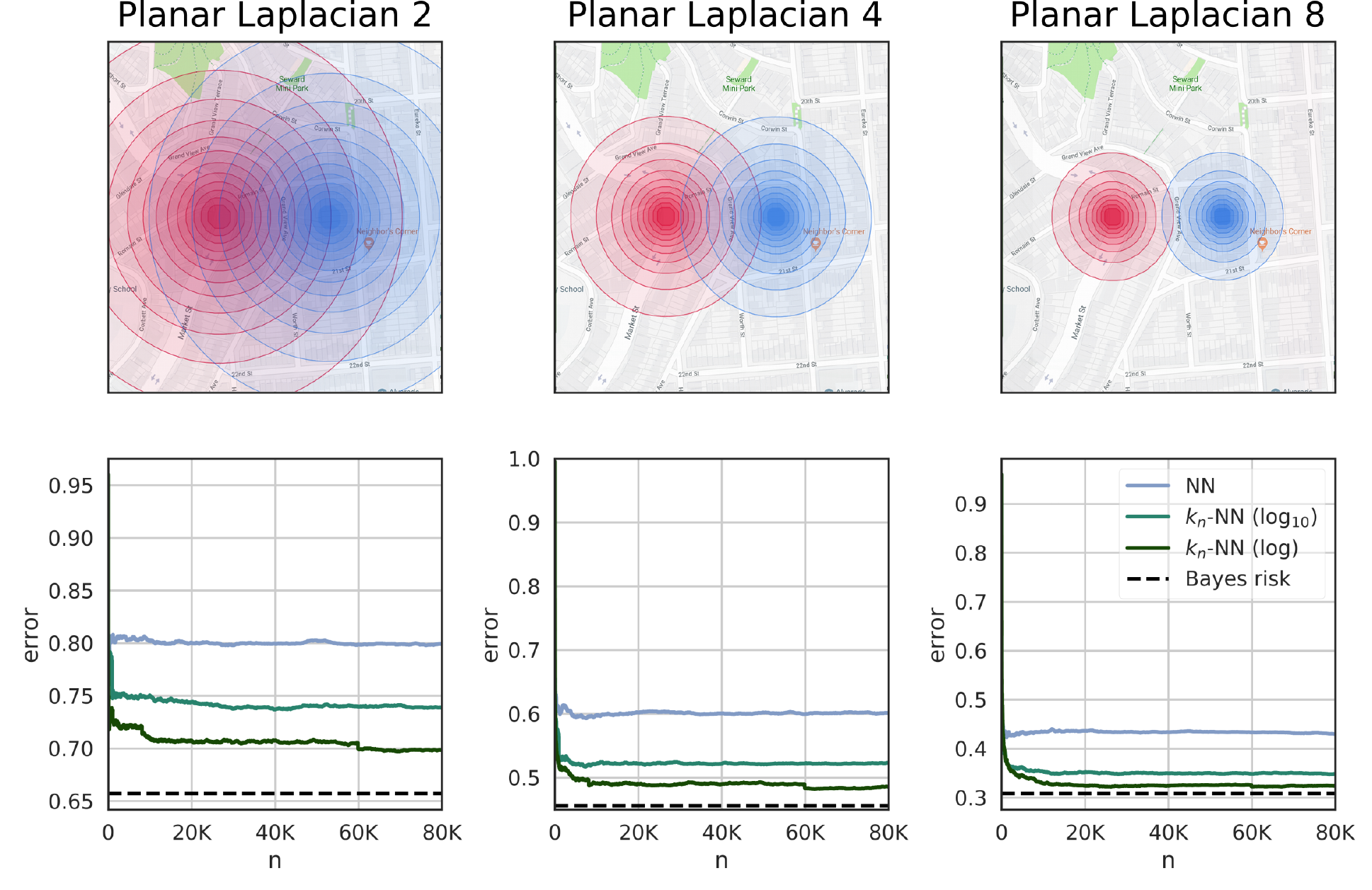}
	\caption{Estimates' convergence speed for the planar Laplacian defense applied to the \gowalla dataset, for  $\nu = 2$, $\nu = 4$ and $\nu = 8$, respectively. Above each graph is represented the  distribution of the geometric noise for two adjacent input cells.}
	\label{fig:laplace-gowalla}
\end{figure*}

The case of the  Blahut-Arimoto mechanism is quite different: 
surprisingly, the output probability concentrates on a small number of locations. For instance, in the case $\nu=2$, 
with 100K sampled pairs we obtained only $19$ different output locations (which reduced to $14$ after we  mapped them on the $20 \times 20$ grid).
Thanks to the small number of actual outputs, all the methods converge very fast. The results are shown in \autoref{fig:arimoto-gowalla}  and in \autoref{tab:arimoto}.

\begin{figure*}
	\centering
	\includegraphics[width=0.65\textwidth]{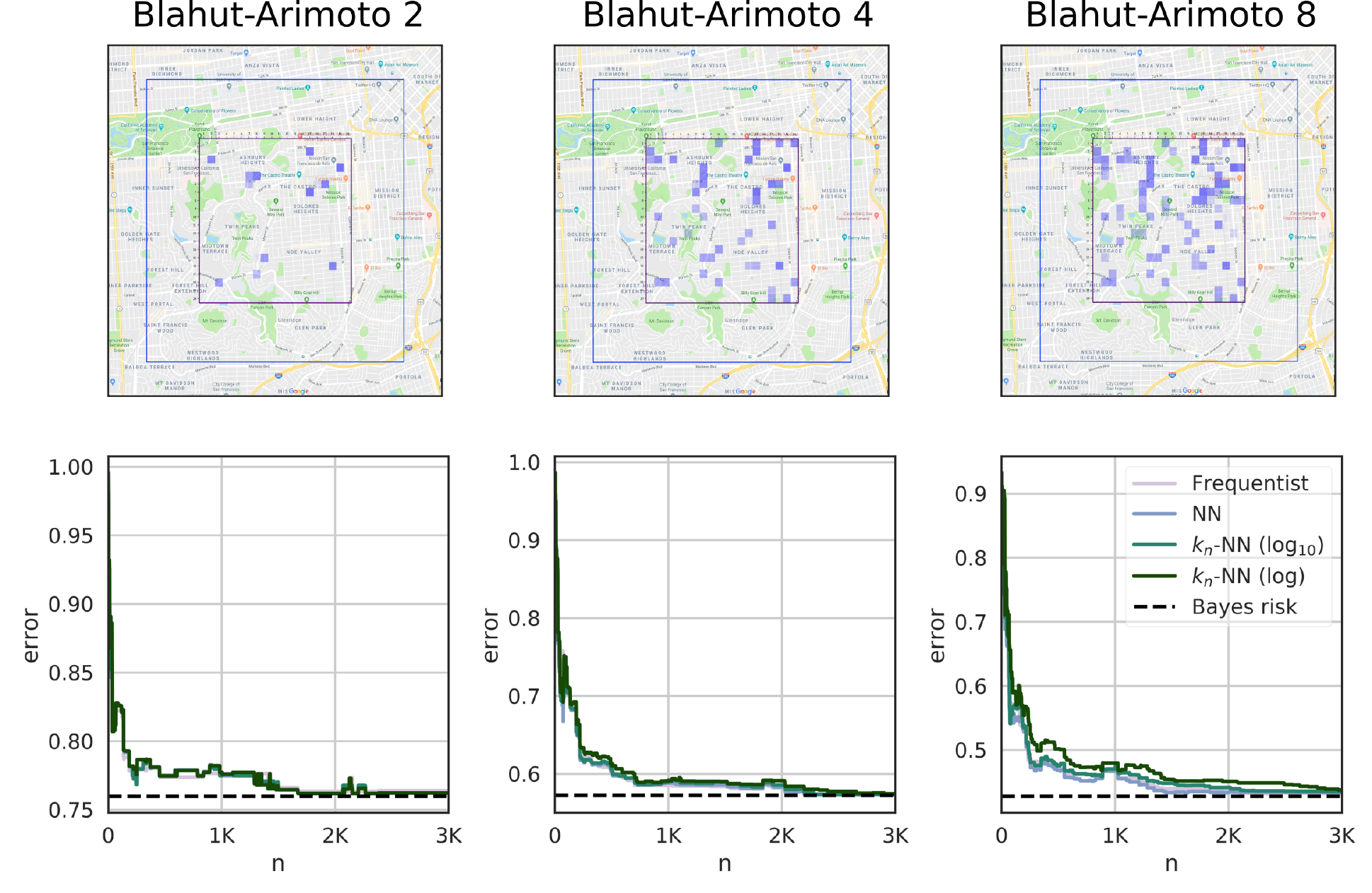}
	\caption{Estimates' convergence speed for the Blahut-Arimoto defense applied to the \gowalla dataset, for  $\nu = 2$, $\nu = 4$ and $\nu = 8$, respectively. Above each graph is represented the  distribution of the output probability as produced by the mechanism. All the outputs with non-null probability turn out to be inside the input square.
		The outputs are points on the $340 \times 340$ grid, but here are  mapped on the coarser $20 \times 20$ grid for the sake of visual clarity. }
	\label{fig:arimoto-gowalla}
\end{figure*}

\begin{table}[!t]
	\centering
	\renewcommand{\arraystretch}{1.3}
	\caption{Convergence 
	 for 
		the Planar Geometric   for various $\nu$.}
	\label{tab:geometric}
	\begin{tabular}{cccccc}
		&          &       &    & \multicolumn{2}{c}{\knn}\\
		\cline{5-6}
		$\nu$ & $\delta$ & frequentist & NN & $\log10$ & $\log$ \\
		\toprule
		2 & 0.1  & X  & X  & 25\,795  & {\bf1\,102} \\
		& 0.05  & X  & X  & X  & {\bf55\,480} \\
		\hline
		4 & 0.1  & X  & X  & 36\,735  & {\bf2\,820} \\
		& 0.05  & X  & X  & X  & {\bf59\,875} \\
		\hline
		8 & 0.1  & X  & X  & 15\,253  & {\bf5\,244} \\
		& 0.05  & X  & X  & X  & {\bf19\,948} \\
		
		\bottomrule
	\end{tabular}
\end{table}

\begin{table}[!t]
	\centering
	\renewcommand{\arraystretch}{1.3}
	\caption{
Convergence 
		for the Planar Laplacian   for various $\nu$.}
	\label{tab:laplace}
	\begin{tabular}{cccccc}
		&          &       &    & \multicolumn{2}{c}{\knn}\\
		\cline{5-6}
		$\nu$ & $\delta$ & frequentist & NN & $\log10$ & $\log$ \\
		\toprule
		2 & 0.1  & N/A  & X  & X  & {\bf259} \\
		\hline
		4 & 0.1  & N/A  & X  & X  & {\bf4\,008} \\
		\hline
		8 & 0.1  & N/A  & X  & X  & {\bf6\,135} \\
		& 0.05  & N/A  & X  & X  & {\bf19\,961} \\
		\bottomrule
	\end{tabular}
\end{table}

\begin{table}[!t]
	\centering
	\renewcommand{\arraystretch}{1.3}
	\caption{Convergence 
		for the Blahut-Arimoto   for various $\nu$.}
	\label{tab:arimoto}
	\begin{tabular}{cccccc}
		&          &       &    & \multicolumn{2}{c}{\knn}\\
		\cline{5-6}
		$\nu$ & $\delta$ & frequentist & NN & $\log10$ & $\log$ \\
		\toprule
		2 & 0.1  & {\bf37}  & {\bf37}  & {\bf37}  & {\bf37} \\
		& 0.05  & {\bf135}  & {\bf135}  & {\bf135}  & {\bf135} \\
		& 0.01  & 1\,671  & 1\,664  & {\bf1\,408}  & {\bf1\,408} \\
		& 0.005  & 6\,179  & 6\,179  & {\bf1\,671}  & {\bf1\,671} \\
		\hline
		4 & 0.1  & {\bf220}  & {\bf220}  & {\bf220}  & 257 \\
		& 0.05  & 503  & {\bf502}  & 509  & 703 \\
		& 0.01  & 2\,029  & {\bf1\,986}  & 2\,055  & 2\,404 \\
		& 0.005  & 2\,197  & {\bf2\,055}  & 2\,280  & 2\,658 \\
		\hline
		8 & 0.1  & {\bf345}  & 401  & 553  & 1\,285 \\
		& 0.05  & 1\,285  & {\bf1\,170}  & 1\,343  & 1\,679 \\
		& 0.01  & 2\,104  & {\bf2\,017}  & 2\,495  & 4\,190 \\
		& 0.005  & {\bf2\,231}  & {\bf2\,231}  & 3\,881  & 6\,121 \\
		\bottomrule
	\end{tabular}
\end{table}

\section{Comparison with \leakiest}\label{sec:comparison}
\leakwatch~\cite{ChothiaKN14:esorics}
and \leakiest~\cite{Chothia:13:CAV} are the major existing black-box leakage
measurement tools, both based on the frequentist approach.
\leakwatch is an extension of \leakiest, which uses the
latter as a subroutine, but \leakiest is more
feature rich:
both tools compute Shannon mutual information (MI) and min-entropy leakage (\meleakage)
on the finite-output case,
but \leakiest can also
perform tests in the continuous output case.
We compare \leakiest with our methods,
for a time side channel
in the RFID chips of the European passports
and for the \gowalla examples of the previous section.

LeakiEst performs two functions: i) a statistical test, detecting
if there is evidence of leakage (here referred to as \textit{leakage evidence test}),
and ii) the estimation of ME (discrete) or MI (discrete and continuous output).
The leakage evidence test generates a ``no leakage'' distribution
via a bootstrapping variant, it estimates the leakage measure on it,
and it compares this estimate with the measure computed on the original
data: if its value is far from the former (w.r.t. some defined confidence
level), then the tool declares there is evidence of leakage.
The second function estimates the distribution with an appropriate
method (frequentist, for finite outputs, Kernel Density Estimation, for continuous
outputs).

\subsection{Time side channel on e-Passports' RFID chips}

\begin{table}[t]
	\renewcommand{\arraystretch}{1.3}
	\caption{Leakage of European  passports
		}
	\label{tab:epassports}
	\centering
	\begin{tabular}{cccc}
		Passport & \leakiest: Leakage evidence & \bleau: \bayesrisk &  \\ \toprule
		British        &           yes            &       0.383        &  \\
		German        &         {\bf no}          &       0.490        &  \\
		Greek         &         {\bf no}          &       0.462        &  \\
		Irish         &            yes            &       0.350        &  \\ \bottomrule
	\end{tabular}\\\vspace{0.1cm}
	Random guessing baseline is $\errorguesspriors=0.5$.
\end{table}

Chothia et al.~\cite{chothia2010traceability} discovered a side-channel attack 
in the way the protocols of various European countries' passports exchanged messages
some years ago. (The protocols have been corrected since then.) 
The problem was that,  upon receiving a message, the e-Passport would 
 first check the Message Authentication Code (MAC),
and only \textit{afterwards} verify the nonce (so to assert the message was not replayed).
Therefore, an attacker who previously intercepted 
a valid message from a legitimate session could replay the message and
detect a difference between  the response time of the victim's passport and any other passport; 
this could be used to track the victim.
As an initial solution, Chothia et al.~\cite{Chothia:13:CAV} proposed to add
padding to the response time, and they used \leakiest to
look for any evidence of leakage after such a defense.

We compared \bleau and \leakiest on
the padded timing data~\cite{epassports.php}.
The secret space contains answers to
the binary question: ``is this the same passport?'';
the dataset is balanced, hence $\errorguesspriors = 0.5$.
We make this comparison on the basis that,
if \leakiest detects no leakage, then the Bayes risk should
be maximum: no leakage happens if and only if
$\bayesrisk = \errorguesspriors$.
We compute ME from the Bayes risk as:
\begin{equation}\meleakage := -\log_2 (1-\errorguesspriors) + \log_2 (1-\bayesrisk) \,.\end{equation}

For \bleau, we  randomly split  the data into training ($75\%$)
and hold-out sets, and then estimated \bayesrisk
on the latter;
we repeated this for $100$ different random
initialization seeds, and averaged the estimates.
Results in \autoref{tab:epassports} show
two cases where \leakiest did not find enough evidence 
of leakage, while \bleau indicates non-negligible leakage.
Note that, because \bleau's results are based on an actual
classifier, they implicitly demonstrate there exists an attack
that succeeds with accuracy 51\% and 54\%.
We attribute this discrepancy between the tools to the fact that the dataset
is small ($\approx$1K examples), and \leakiest may not find
enough evidence to reject the hypothesis of
``no leakage''; indeed, \leakiest sets a fairly high
standard for this decision (95\% confidence interval).

\subsection{\gowalla dataset}

\begin{table}[!t]
	\renewcommand{\arraystretch}{1.3}
	\caption{Estimated leakage of privacy mechanisms on \gowalla data
		}
	\label{tab:leakiest-gowalla}
	\centering
	\begin{tabular}{cccccc}
				&         & \multicolumn{2}{c}{\leakiest} & &\\
	Mechanism   &  $\nu $ & L.E. & \meleakage & {\bleau}: \meleakage  &  True \meleakage \\	\toprule
	B.-Arimoto  &   2   &    no* & 1.481         & 1.479                   & 1.501     \\
				&   4   &    no* & 2.305         & 2.310                   & 2.304           \\
				&   8   &    no* & 2.738         & 2.746                   & 2.738           \\
	Geometric   &   2   &    {\bf no}  & 2.585         & 1.862                   & 1.988               \\
				&   4   &    {\bf no}  & 2.859         & 2.591                   & 2.638               \\
				&   8   &    {\bf no}  & 3.105         & 2.983                   & 2.996           \\
\end{tabular}\\
\vspace{0.1cm}
\begin{tabular}{cccccc}
	Mechanism & $\nu$ & \leakiest: L.E.  & \bleau: \meleakage  & True \meleakage \\ \toprule
	Laplacian &   2   &  {\bf no}  & 1.802   &  1.987  \\
			  &   4   &  {\bf no}  & 2.550   &  2.631                  \\
			  &   8   &  {\bf no}  & 2.970   &  3.003                  \\ \bottomrule
	\end{tabular}\\\vspace{0.1cm}
	L.E.: leakage evidence test
\end{table}

We compare \bleau and   \leakiest on the location privacy mechanisms (\autoref{sec:gowalla}):
Blahut-Arimoto,  planar Geometric, and planar Laplacian. 
The main interest is to verify whether the advantage of \bleau w.r.t.
the frequentist approach, which we observed for large output spaces,
translates into an advantage also  w.r.t. \leakiest. 
For the first two mechanisms we also compare the \meleakage estimates.
For the Laplacian case (continuous), we only use \leakiest's leakage evidence test.

We run \bleau  and \leakiest on the datasets as in \autoref{sec:gowalla}.
Results in \autoref{tab:leakiest-gowalla} show that, in the cases of
planar Geometric and Laplacian distributions, \leakiest does not detect
any leakage (the tool reports ``Too small sample size'');
furthermore, the \meleakage estimates it provides for the
planar Geometric distribution are far from their true values.
\bleau, however, is able to produce more reliable estimates.

The Blahut-Arimoto results are more interesting:
because of the  small number of actual outputs,
the \meleakage estimates of \bleau and \leakiest perform
equally well.
However,
even in this case \leakiest's
leakage evidence test reports
``Too small sample size''.
We think the reason is that \leakiest takes into account the
declared
size of the object space
rather then the effective number of observed individual outputs;
this problem should be easy to fix by inferring
the output size from the examples (this is the meaning of the ``*'' in \autoref{tab:leakiest-gowalla}).

\section{Conclusion and Future Work}

We showed that the black-box leakage of a system, measured
until now with classical statistics paradigms (\textit{frequentist} approach),
can be effectively estimated via ML techniques.
We considered a set of such techniques based on the nearest neighbor
principle (i.e., close observations should be assigned the same
secret), and evaluated them thoroughly on synthetic and real-world
data.
This allows to tackle problems that were impractical until
now;
furthermore, it sets a new paradigm in black-box security:
thanks to an equivalence
between ML and black-box
leakage estimation, many results from the ML theory can
be now imported into this practice (and vice versa).

Empirical evidence shows that
the nearest neighbor techniques we introduce
excel
whenever there is a notion of metric they can exploit in the output space:
whereas for unseen observations the frequentist approach needs to take
a guess,
nearest neighbor methods can
use the information of neighboring observations.
We also observe that whenever the output distribution
is irregular,
they are equivalent to the frequentist
approach, but for maliciously crafted systems
they can be misled.
Even in those cases, however, we remark that asymptotically
they are equivalent to the frequentist approach,
thanks to their universal consistency
property.

We also indicated that, as a consequence of the No
Free Lunch (NFL) theorem in ML, no estimate
can guarantee optimal convergence.
We therefore proposed
\bleau, a combination of frequentist and nearest neighbor
rules, which runs all these techniques on a
system, and selects the estimate that converges
faster.
We expect this work will inspire researchers
to explore new leakage estimators from the ML literature;
in particular, we showed that any ``universally
consistent'' ML rule can be used to estimate
the leakage of a system.
Future work may focus on other rules from which
one can obtain universal consistency
(e.g., Support Vector Machine (SVM) and neural networks);
we discuss this further in Appendix~\ref{appendix:more-tools}.

A crucial advantage of the ML formulation, as opposed
to the classical approach, is
that it gives immediate guarantees for systems with a continuous
output space.
Future work may extend this to systems
with continuous secret space, which in ML terms would
be formalized as regression (as opposed to the classification
setting we considered here).

A current limitation of our methods is that
they do not provide confidence intervals.
We leave this as an open problem.
We remark, however, that for continuous systems
it is not possible
to provide confidence intervals
(or to prove convergence rates) under our weak assumptions~\cite{devroye2013probabilistic};
this constraint applies to any leakage
estimation method.

We reiterate, however, the great advantage of
ML methods: they allow tackling systems for which
until now we could not measure security, with a
strongly reduced number of examples.

\ifanonymous
\else

\ifCLASSOPTIONcompsoc
  \section*{Acknowledgments}
\else
  \section*{Acknowledgment}
\fi

Giovanni Cherubin has been partially supported by an EcoCloud grant.
The work of Konstantinos Chatzikokolakis and Catuscia Palamidessi was partially supported by the 
ANR project REPAS. 

We are very thankful to Marco Stronati, who was initially
involved in this project, and without whom the
authors would have not
started working together.
We are grateful to Tom Chothia and Yusuke Kawamoto, for their 
help to understand \leakiest.
We also thank Fabrizio Biondi for useful discussion. 

This work began as a research visit whose (secondary) goal
was for some of the authors and Marco to climb in Fontainebleau, France.
The trip to the magic land of Fontainebleau never happened --
although climbing has been a fundamental aspect of the collaboration;
the name \bleau is
to honor this unfulfilled dream.
We hope one day the group will
reunite,
and finally climb there together.

\fi

\bibliographystyle{IEEEtran}
\bibliography{IEEEabrv,bibliography,biblio_cat}

\appendices

\section{Additional tools from ML}
\label{appendix:more-tools}

The ML literature can offer several more tools to
black-box security.
We now enumerate additional UC rules, a lower bound
of the Bayes risk,
and a general way of obtaining estimates that
converge from below.

The family of UC rules is fairly large.
An overview of them is by Devroye et al.~\cite{devroye2013probabilistic},
who, in addition to nearest neighbor methods,
report histogram rules and kinds of neural networks;
these are UC under requirements on their parameters.
Steinwart proved that Support Vector Machine (SVM) is
also UC for some parameter choices,
in the case $|\secretspace|=2$~\cite{steinwart2002support};
to the best of our knowledge, attempts to construct an SVM that is
UC when $|\secretspace| > 2$ have failed so far (e.g., \cite{glasmachers2010universal})

In applications with
strict security requirements, a (pessimistic) lower bound of the
Bayes risk may be desirable.
From a result by
Cover and Hart
one can derive a lower bound on the Bayes risk based
on the NN error, $R^{NN}$~\cite{cover1967nearest}:
as $n \rightarrow \infty$:
\begin{equation}\frac{|\secretspace|-1}{|\secretspace|} \left(1-\sqrt{1-\frac{|\secretspace|}{|\secretspace|-1}R^{NN}}\right) \leq \bayesrisk \,.\end{equation}

This was used as the basis for measuring the black-box leakage of
website fingerprinting defenses~\cite{cherubin2017bayes}.

Finally, one may obtain estimators that converge
to the Bayes risk in expectation {\em from below},
for example, by estimating
the error of a \knn rule on its training set~\cite{fukunaga1987bayes,cherubin2017bayes}.

\section{Description of the synthetic systems}
\label{appendix:systems}

\subsection{Geometric system}
Geometric systems are typical in differential privacy and are obtained by adding negative exponential noise to the result of a query. 
The reason is that the property of DP is expressed in terms of a factor between the probability of a reported answer, and that of its immediate neighbor. 
A similar construction holds for the geometric mechanism implementing geo-indistinguishability. In that case the noise is added to the real location to report an obfuscated location. 
Here we give an abstract definition of a geometric system, in terms of secrets (e.g., result of a query / real location) and observables (e.g., reported answer / reported location).

Let $\secretspace$ and $\objectspace$ be sets of consecutive natural numbers,  with the standard notion of distance.
Two numbers  $s,s'\in \secretspace$ are called \emph{adjacent} if $s=s'+1$ or  $s'= s+1$.

Let $\nu$ be a real non-negative number and consider a 
function $g: \secretspace \mapsto \objectspace$.
After adding negative exponential noise to the output of $g$, the resulting geometric system is described by 
the following channel matrix:
\begin{equation}
\channel_{\secret, \object} = P(o \mid s) =
\lambda \exp\left(- \nu {\mid g(s) - o \mid}\right) \,,
\end{equation}
where $\lambda$ is a normalizing factor.
Note that the privacy level is defined by $\nicefrac{\nu}{\Delta_g}$, 
where $\Delta_g$ is the sensitivity of $g$:
\begin{equation}\Delta_g = \max_{\secret_1\sim\secret_2\in \secretspace}(g(\secret_1)-g(\secret_2)) \,,\end{equation}
where $\secret_1\sim\secret_2$ means $\secret_1$ and $\secret_2$ are adjacent.
Now let $\secretspace = \{1,\ldots, w\}$, $\objectspace = \{1, ..., w'\}$,
we select $g$ to be $g(s) = s\cdot w'/w$.
We define
\begin{equation}
\lambda =  \begin{cases}
e^\nu/(e^\nu + 1) \quad & \mbox{if $o=1$ or $o=w'$}\\
(e^\nu - 1)/(e^\nu + 1) \quad & \mbox{otherwise}\,,
\end{cases}
\end{equation}
so to truncate the distribution at its boundaries.

This definition of Geometric system prohibits the case
$|S|>|O|$.
To consider such case, we generate a repeated geometric
channel matrix, such that
\begin{equation}
\channel'_{\secret, \object} = \channel_{\secret', \object} \quad \secret' = \secret \bmod |O| \,,
\end{equation}
where $\channel_{\secret, \object}$ is the geometric
channel matrix described above.

\subsection{Multimodal geometric system}
We construct a multimodal distribution as the weighted sum 
of two geometric distributions, shifted by some shift
parameter.
Let $\channel_{\secret, \object}$ be a geometric
channel matrix.
The respective multimodal geometric channel, for shift
parameter $\sigma$, is:
\begin{equation}\channel^M_{\secret, \object} = w_1\channel_{\secret, \object} +
	w_2\channel_{\secret + 2\sigma, \object} \,.\end{equation}

In experiments, we used $\sigma=5$ and weights
$w_1=w_2=0.5$.

\subsection{Spiky system}
Consider an observation space constituted of $q$ consecutive integer numbers 
$\objectspace = \{0, ..., q-1\}$, for some even positive integer $q$,
and  secrets space $|\secretspace| = 2$. 
Assume that $\objectspace$  is a ring with the operations $+$ and $-$ 
defined as the sum  and the difference modulo $q$, respectively, 
and  consider the distance on $\objectspace$ defined as:
$d(i,j)=|i-j|$. (Note that $(\objectspace,d)$ is a ``circular'' structure,
that is, $d(q-1,0)=1$.)
The Spiky system has uniform prior,
and channel matrix constructed as follows:
\begin{equation}
\channel_{\secret, \object} =
\begin{bmatrix}
\nicefrac{2}{q} & 0 & \nicefrac{2}{q} & \dots & 0\\
0 & \nicefrac{2}{q} & 0 & \dots & \nicefrac{2}{q}
\end{bmatrix} \,.
\end{equation}

\section{Detailed convergence results}
\label{appendix:convergence-results}

Convergence for multimodal geometric systems,
when varying $\nu$ and for fixed
$|\secretspace| \times |\objectspace| = 100 \times 10K$.

\begin{table}[ht!]
	\centering
	\begin{tabular}{lccccc}
		&          &       &    & \multicolumn{2}{c}{\knn}\\
		\cline{5-6}
		System & $\delta$ & Freq. & NN & $\log_{10}$ & $\log$ \\
		\toprule

		\multirow{3}{5em}{{\bf Multimodal\\$\bm \nu$ = 1.0}}
		& 0.1 & 3\,008 & \textbf{369} & 478 & 897\\
		& 0.05 & 5\,938 & \textbf{495} & 754 & 1\,267\\
		& 0.01 & 26\,634 & \textbf{765} & 1\,166 & 1\,487\\
		& 0.005 & 52\,081 & \textbf{765} & 1\,166 & 1\,487\\

		\vspace{-0.3em}\\
		\multirow{3}{5em}{{\bf Multimodal \\$\bm \nu$ = 0.1}}
		& 0.1 & 24\,453 & \textbf{398} & 554 & 821\\
		& 0.05 & 44\,715 & \textbf{568} & 754 & 1\,175\\
		& 0.01 & 149\,244 & 4\,842 & \textbf{1\,166} & 1\,487\\
		& 0.005 & 226\,947 & 79\,712 & \textbf{1\,166} & 1\,487\\
		
		\vspace{-0.3em}\\
		\multirow{3}{5em}{{\bf Multimodal \\$\bm \nu$ = 0.01}}
		& 0.1 & 27\,489 & 753 & 900 & \textbf{381}\\
		& 0.05 & 103\,374 & 101\,664 & 92\,181 & \textbf{31\,452}\\
		~\\
		
		\bottomrule
	\end{tabular}
\end{table}

Detailed convergence results for a Spiky system of size
$|\secretspace|=2$ and $|\objectspace|=10K$.

\begin{table}[ht!]
	\centering
	\begin{tabular}{lcccc}
		&          &       &    \multicolumn{2}{c}{\knn}\\
		\cline{4-5}
		$\delta$ & Freq. & NN & $\log_{10}$ & $\log$ \\
		\toprule
		
		0.1 & \textbf{15\,953} & 22\,801 & 52\,515 & 99\,325\\
		0.05 & \textbf{22\,908} & 29\,863 & 62\,325 & 112\,467\\
		0.01 & \textbf{38\,119} & 44\,841 & 81\,925 & 137\,969\\
		0.005 & \textbf{44\,853} & 51\,683 & 91\,661 & 147\,593\\
		\bottomrule
	\end{tabular}
\end{table}

\section{Uniform system}

We measured convergence of the methods for a uniform system;
this system is constructed so that all secret-object examples
are equally likely, that is $\mu(s, o) = \mu(s', o')$
for all $s, o \in \examplespace$.
The Bayes risk in this case is $\bayesrisk = 1 - \nicefrac{1}{|\secretspace|}$.

\autoref{fig:uniform} shows that even in this case all rules
are equivalent.
Indeed, because the system leaks nothing about its secrets,
all the estimators need to randomly guess; but because for this
system the Bayes risk is identical to random guessing error
($\bayesrisk = 1 - \nicefrac{1}{|\secretspace|} = \errorguesspriors$),
all the estimators converge immediately to its true value.

\begin{figure}[!t]
	\centering
	\includegraphics[width=0.25\textwidth]{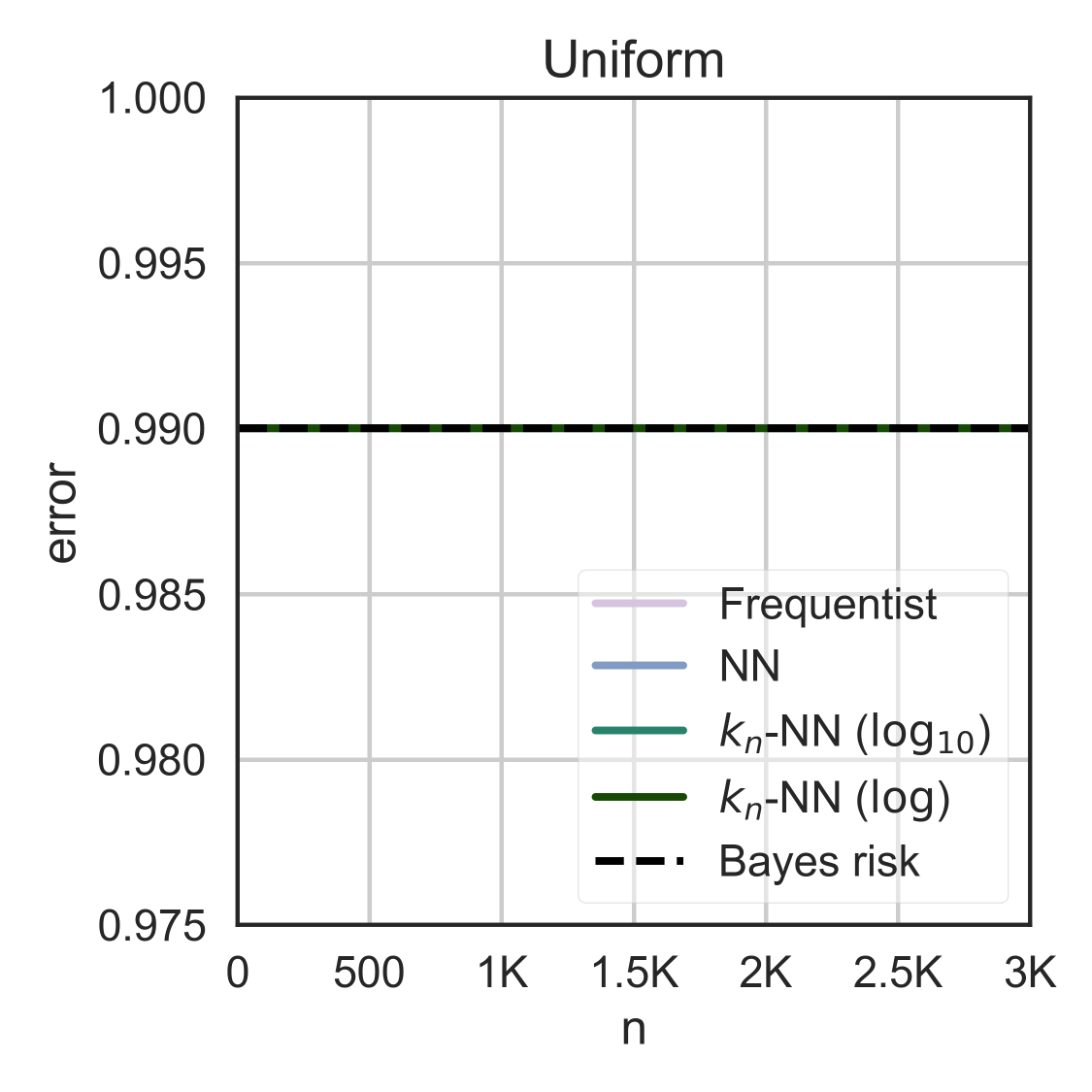}
	\caption{Convergence for a Uniform system of size $100\times100$.}
	\label{fig:uniform}
\end{figure}

\section{Approximation of the frequentist estimate}
\label{appendix:frequentist}

\begin{figure}[!t]
	\centering
	\includegraphics[width=0.225\textwidth]{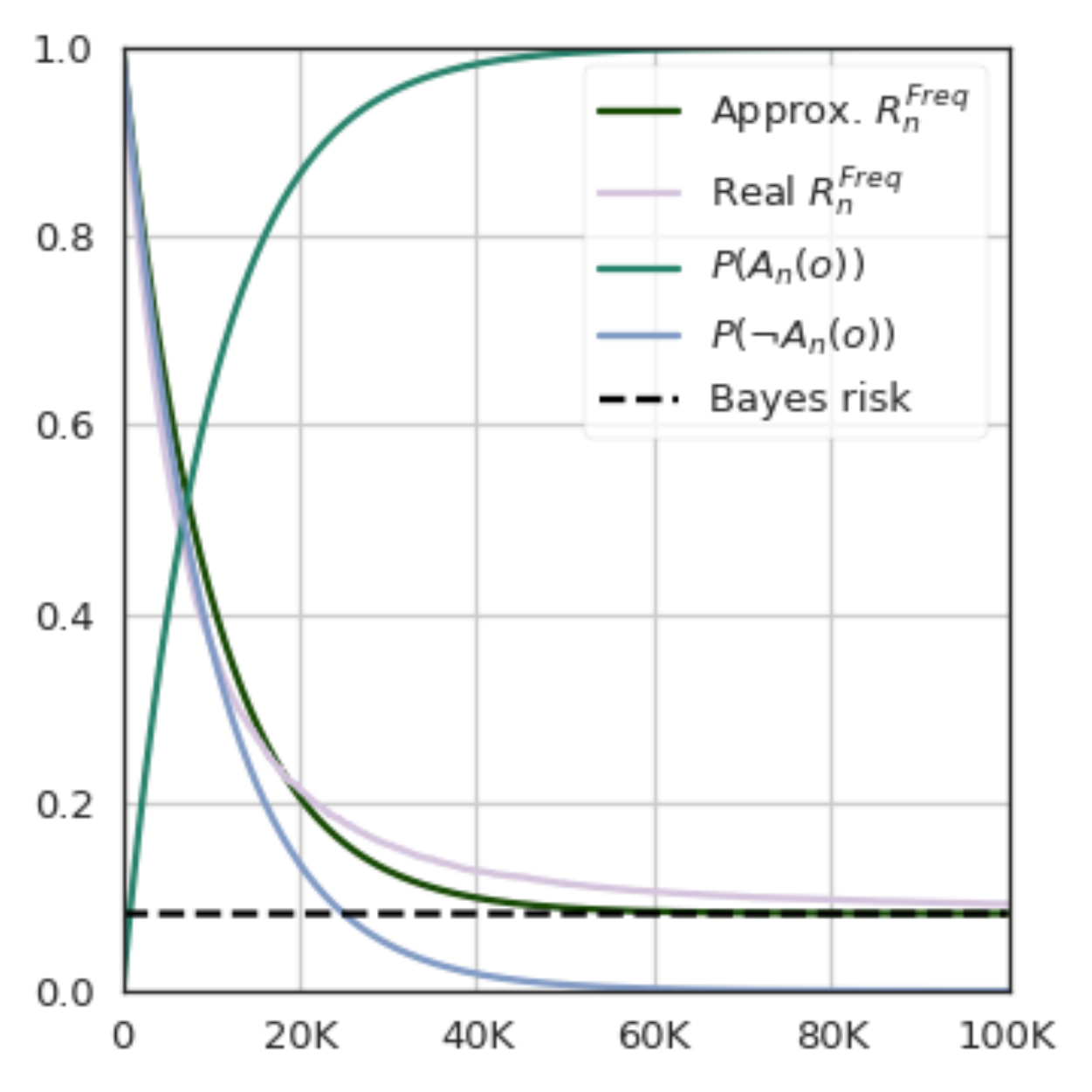}
	\caption{Approximation of the frequentist estimate as $n$ grows for $\bayesrisk \approx 0.08$, $|O|=10K$, and $|S|=1K$;
		the approximation is compared with the real frequentist estimate $R^{Freq}_n$.}
	\label{fig:frequentist-approx}
\end{figure}

To better understand the behavior of the frequentist approach
for observations that were not in the training data,
we derive a crude approximation of this estimate in terms
of the size of training data $n$.
The approximation makes the following assumptions:
\begin{enumerate}
	\item each observation $\object \in \objectspace$ is equally 
	likely to
	appear in training data (i.e., $P(\object) = 1 - \frac{1}{|\objectspace|}$);
	\item if an observation appears in the training data, the frequentist
	approach outputs the secret minimizing the Bayes risk;
	\item the frequentist estimate knows the real priors $\priors$;
	\item if an observation does not appear in the training data, then the frequentist
		approach outputs the secret with the maximum prior probability.
\end{enumerate}
The first two assumptions are very strong, and thus this is just
an approximation of the real trend of such estimate.
However, in practice it approximates well the real trend
\autoref{fig:frequentist-approx}.

Let $A_n(o)$ denote the event ``observation $o$ appears in a training set
of $n$ examples''; because of assumption 1),
$P(A_n(o)) = 1 - \left(1 - \frac{1}{|\objectspace|}\right)^n$.
The conditional Bayes risk estimated with a frequentist approach given
$n$ examples is:
\begin{align*}
r_n(o) = & r_n(o | A_n(o))P(A_n(o)) + r_n(o | \neg A_n(o))P(\neg A_n(o)) = \\
= & \left(1 - \max_{s \in \secretspace} \frac{\hat{\channel}_{s,o}\hat{\priors}(s)}{P(o)}\right) P(A_n(o)) +\\
& + (1 - \max_{s \in \secretspace} \hat{\priors}(s)) P(\neg A_n(o)) \approx \\
\approx & \left(1 - \max_{s \in \secretspace} \frac{\channel_{s,o}\prior{s}}{P(o)}\right) P(A_n(o)) +\\
& + (1 - \max_{s \in \secretspace} \prior{s}) P(\neg A_n(o))
\end{align*}
Assumptions 2) and 3) were used in the last step.
From this expression, we derive the frequentist estimate of \bayesrisk t step
$n$:
\begin{align*}
R^{\mathit{Freq}}_n = \;& \mathbb{E}r_n =\\
= \;& \sum_{o \in \objectspace} P(o) \left(1 - \max_{s \in \secretspace} \frac{\channel_{s,o} \prior{s}}{P(o)}\right) P(A_n(o)) +\\
& + \sum_{o \in \objectspace} P(o) (1 - \max_{s \in \secretspace} \prior{s}) P(\neg A_n(o)) =\\
= \;& P(A_n(o)) \left( \sum_{o \in \objectspace} P(o) - \sum_{o \in \objectspace} \max_{s \in \secretspace} \channel_{s,o}\prior{s}) \right) +\\
& + P(\neg A_n(o)) (1 - \max_{s \in \secretspace} \prior{s}) \sum_{o \in \objectspace} P(o) =\\
= \;& P(A_n(o)) \left(1 - \sum_{o \in \objectspace} \max_{s \in \secretspace} \channel_{s,o}\prior{s})\right) +\\
& + P(\neg A_n(o)) (1 - \max_{s \in \secretspace} \prior{s}) =\\
= \;& P(A_n(o)) \bayesrisk + P(\neg A_n(o)) \errorguesspriors =\\
= \;& \bayesrisk \left(1 - \left(1 - \frac{1}{|\objectspace|}\right)^n\right) +
\errorguesspriors\left(1 - \frac{1}{|\objectspace|}\right)^n \,.
\end{align*}
Note that in the second step we used $P(A_n(o))$ as a constant, which is
allowed by assumption 1).

The expression of $R_n$ indicates that $P(A_n(o))$ weights between
random guessing according to priors-based random guessing and the Bayes risk; when
$P(A_n(o)) \geq P(\neg A_n(o))$, which happens for
$n \geq - \frac{\log{2}}{\log\left(1 - \frac{1}{|\objectspace|}\right)}$
the frequentist approach starts approximating using the
actual Bayes risk (\autoref{fig:frequentist-approx}).

\section{\gowalla details}
We report in \autoref{tab:gowalla} the
real Bayes risk estimated analytically for the \gowalla dataset
defended using the various mechanisms, and their respective
utility.

\begin{table}[ht]
	\centering
	\label{tab:true-bayes-gowalla}
	\caption{True Bayes risk and utility for \gowalla dataset
		defended using various location privacy mechanisms.}
	\renewcommand{\arraystretch}{1.3}
	\begin{tabular}{lccc}
		Mechanism & $\nu$ & \bayesrisk & Utility \\ \toprule
		Blahut-Arimoto        &   2   &   0.760    & 334.611 \\
		                      &   4   &   0.571    & 160.839 \\
		                      &   8   &   0.428    & 96.2724 \\ \hline
		Geometric             &   2   &   0.657    & 288.372 \\
		                      &   4   &   0.456    & 144.233 \\
		                      &   8   &   0.308    & 96.0195 \\ \hline
		Laplacian             &   2   &   0.657    & 288.66  \\
		                      &   4   &   0.456    & 144.232 \\
		                      &   8   &   0.308    & 96.212  \\ \bottomrule
	\end{tabular}
	\label{tab:gowalla}
\end{table}

\section{Application to time side channel}
\label{appendix:exponentiation}

\begin{table}[!t]
	\renewcommand{\arraystretch}{1.3}
	\caption{Number of unique secrets and observations for the time
		side channel to finite field exponentiation.}
	\label{tab:size-finite-field}
	\centering
	\begin{tabular}{ccc}
		Operands' size & $|\secretspace|$ & $|\objectspace|$ \\
		\toprule
		4 bits & $2^4$ & 34\\
		6 bits & $2^6$ & 123\\
		8 bits & $2^8$ & 233\\
		10 bits & $2^{10}$ & 371\\
		12 bits & $2^{12}$ & 541\\
		\bottomrule
	\end{tabular}
\end{table}

We use \bleau to measure the leakage in the running time
of the {\em square-and-multiply} exponentiation
algorithm in the finite field $\mathbb{F}_{2^w}$;
exponentiation in $\mathbb{F}_{2^w}$ is relevant, for example,
for the implementation of the ElGamal cryptosystem.

We consider a hardware-equivalent implementation
of the algorithm computing $m^s$ in $\mathbb{F}_{2^w}$.
We focus our analysis
on the simplified scenario of
a ``one-observation'' adversary,
who makes exactly {\em one} measurement
of the algorithm's execution time $o$, and aims to predict the
corresponding secret key $s$.

A similar analysis was done by Backes and K{\"o}pf~\cite{backes2008formally}
by using a leakage estimation method based on the frequentist approach.
Their analysis also extended to a ``many-observations adversary'',
that is, an adversary
who can make $m$ observations $(o_1, ..., o_m)$, all generated from
the same secret $s$, and has to predict $s$ accordingly.

\subsection{Side channel description}
Square-and-multiply is a fast algorithm for computing
$m^s$ in the finite field $\mathbb{F}_{2^w}$, where $w$ here represents
the bit size of the operands $m$ and $s$.
It works by performing a series of multiplications according
to the binary representation of the exponent $s$,
and its running time is proportional to the number of
1's in $s$.
This fact was noticed by Kocher~\cite{kocher1996timing}, who
suggested side channel attacks to the RSA cryptosystem
based on time measurements.

\subsection{Message blinding}
\balance
We assume the system implements {\em message blinding},
a technique which hides to an adversary the value $m$
for which $m^s$ is computed.
Blinding was suggested as a method for thwarting time side channels~\cite{kocher1996timing},
which works as follows.
Consider, for instance, decryption for the RSA cryptosystem:
$m^d (mod N)$, for some decryption key $d$;
the system first computes $m \cdot r^e$, where $e$
is the encryption key and $r$ is some random value;
then it computes $(mr^e)^d$, and returns the decrypted
message after dividing the result by $r$.

Message blinding has the advantage of hiding information
to an adversary; however, it was shown that it is
not enough for preventing time side channels (e.g.,~\cite{backes2008formally}).

\subsection{Implementation and results}

\begin{figure*}[!bht]
	\centering
	\includegraphics[width=0.75\textwidth]{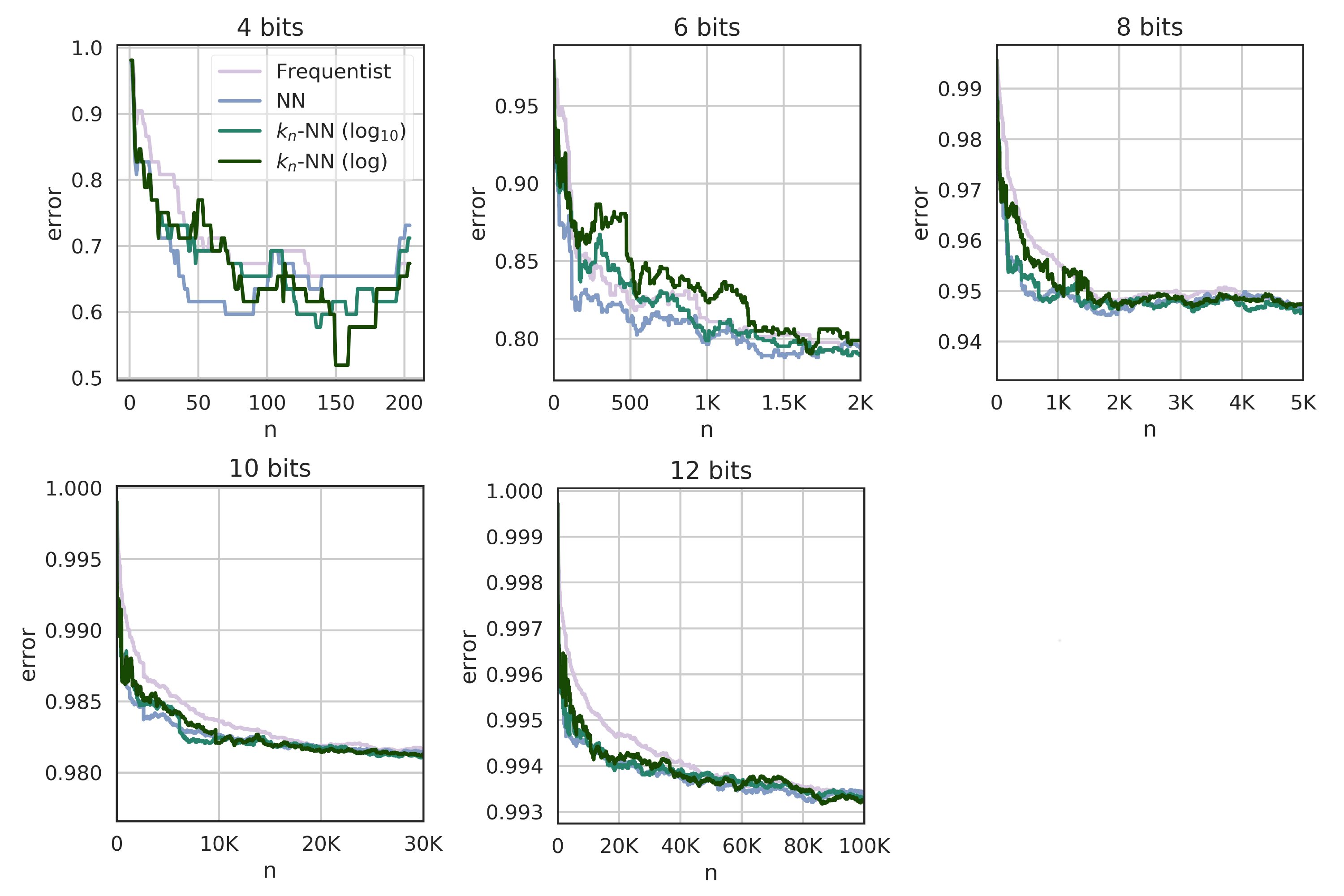}
	\caption{Convergence of the estimates
		for the time side channel attack to the exponentiation algorithm
		as the bit size of the operands increases.}
	\label{fig:time-side-channel}
\end{figure*}

We consider a \texttt{Gezel} implementation of finite field exponentiation.
\texttt{Gezel} is a description language for clocked hardware,
equipped with a simulation environment
whose executions preserve the corresponding circuit's timing information.
This means that the time measurements (i.e., clock cycles) we make
reflect the corresponding circuit implementation~\cite{kopf2006timing}.

We compare the performances of the frequentist and nearest neighbor approaches
in terms of the number of black-box examples required for convergence.
For each bit size $w \in \{4, 6, .., 12\}$, and for all the values
$(m_i, s_i) \in \{0, ..., 2^w-1\}^2$, we run the
exponentiation algorithm to compute $m^s$, and measure its
execution time $o_i$.
As with our application to location privacy (\autoref{sec:gowalla}),
we estimate the
Bayes risk by training a classifier on a set of increasing examples
$n$ and by computing its error on a hold-out set.
We set the size of the hold-out set to
$\min(0.2\cdot2^{2w}, 250\,000)$.

Results in \autoref{fig:time-side-channel}
show that, while for small bit sizes the frequentist approach outperforms
nearest neighbor rules, as $w$ increases, the frequentist approach
requires a much larger number of examples.
Nevertheless, in these experiments we did not notice a substantial
advantage in nearest neighbor rules, even though the output
space is equipped with a notion of metric.
\autoref{tab:size-finite-field} helps interpreting this result:
for larger bit sizes $w$ of the exponentiation operands, the
possible output values (i.e., clock cycles) only increase minimally;
this confirms that, as noticed in our previous experiments, nearest
neighbor and frequentist estimates tend to perform similarly
for systems with small output space.

\end{document}